\documentclass{article}

\usepackage{microtype}
\usepackage{graphicx}
\usepackage{enumitem}
\usepackage{subcaption}
\usepackage{booktabs} 

\usepackage{hyperref}

 \usepackage[preprint]{icml2026}

\usepackage{amsmath}
\usepackage{amssymb}
\usepackage{mathtools}
\usepackage{amsthm}

\usepackage[capitalize,noabbrev]{cleveref}

\theoremstyle{plain}
\newtheorem{theorem}{Theorem}[section]

\theoremstyle{definition}
\newtheorem{definition}[theorem]{Definition}

\theoremstyle{remark}

\usepackage[textsize=tiny]{todonotes}

\icmltitlerunning{Bernstein--Vazirani Networks: Quantum Machine Learning by Interference}

\usepackage{wrapfig}
\usepackage{braket}
\usepackage{tikz}
\usepackage{quantikz}
\usepackage{tikz-cd}
\usetikzlibrary{shapes,positioning,calc, shapes.geometric}
\usepackage{bm}
\usepackage{cuted}

\newcommand{\R}{\mathbb{R}}

\newcommand{\Z}{\mathbb{Z}}
\newcommand{\C}{\mathbb{C}}

\newcommand{\mcal}[1]{\mathcal{#1}}

\newcommand{\Norm}[1]{\|#1\|}

\usepackage{booktabs}    
\usepackage{makecell}    
\usepackage{array}       
\usepackage{multicol,multirow}

\definecolor{quantumrow}{RGB}{230,240,255} 

\crefname{appendix}{Appendix}{Appendices}
\Crefname{appendix}{Appendix}{Appendices}

\begin{document}

\twocolumn[
\begin{center}
\icmltitle{Bernstein--Vazirani Networks: \\
Quantum Machine Learning by Interference}

\vspace{0.0em}

{\large
    Natacha Kuete Meli$^{1}$ \quad
    Tolga Birdal$^{2}$ \quad
    Prayag Tiwari$^{3}$ \quad
    Vladislav Golyanik$^{4, \dagger}$ \quad
    Michael Moeller$^{1, \dagger}$
    \par}

\vspace{0.7em}

{\normalsize\itshape
\noindent
\begin{minipage}[t]{0.48\textwidth}
\centering
    $^{1}$Department of Vision and Graphics,\\
    University of Siegen

    \vspace{0.4em}

    $^{3}$School of Information Technology,\\
    Halmstad University
\end{minipage}
\hfill
\begin{minipage}[t]{0.48\textwidth}
\centering
    $^{2}$Department of Computing,\\
    Imperial College London

    \vspace{0.4em}

    $^{4}$4D and Quantum Vision, \\
    Max Planck Institute for Informatics
\end{minipage}
\par}

\vspace{1em}

\end{center}
]

\begin{abstract}
We introduce Bernstein--Vazirani Networks (BVNs), a non‑variational quantum machine learning framework that leverages quantum interference for supervised learning, demonstrated on vision and representation learning tasks.
In their standard form, BVNs follow the principle of quantum Fourier sampling: 
labelled data are placed in superposition and interfered in the Fourier basis to extract globally informative features. 
We then define generalised BVNs that enable interference in problem-adapted bases, yielding more expressive models under the same measurement budget as in the standard setting.
BVNs achieve universal function approximation through (over)complete interference bases, while training of BVNs is gradient-free. 
Experiments on synthetic and real‑world classification tasks, as well as implicit image representation, show strong generalisation capabilities and competitive performance with classical and quantum baselines\footnote{Project page: \mbox{\url{https://4dqv.mpi-inf.mpg.de/BVNs/}}}. 
\renewcommand{\thefootnote}{\textdagger}
\footnotetext{denotes equal advising}
\end{abstract}
\section{Introduction}
\label{sec:introduction}

Quantum Machine Learning (QML) \cite{biamonte2017quantum,cerezo2022challenges} aims to exploit the high-dimensional Hilbert spaces of quantum systems coupled with quantum effects to build expressively powerful learning models. 
Most existing QML approaches typically rely on Parametrised Quantum Circuits (PQCs) and face significant architectural and trainability challenges. 
A central challenge is the mismatch between classical neural networks’ reliance on nonlinear activations and the linear, unitary dynamics of closed quantum systems, where nonlinearity is recoverable only through Boolean encodings, as noted by Herman et al.~\cite{herman2023expressivity}. 
In addition, PQC optimisation is often hindered by barren plateaus, i.e., regions of the loss landscape where gradients vanish exponentially with system size \cite{mcClean2018barren}, and, even outside such regimes, by the complexity of certified gradient-evaluation methods (e.g., parameter-shift rules \cite{schuld2019evaluating}) combined with measurement noise, imposing substantial overhead in
the execution and sampling of quantum circuits. 
%
A recent study \cite{recio2025train} further confirms this crisis in PQCs, \emph{estimating to 38 years the runtime of a single epoch of gradient descent with the parameter-shift rule on MNIST with a $10^4$‑parameter PQC, requiring ${\sim}10^{15}$ measurement shots.} 
These recent insights expose a tangible gap between the theoretical expressivity of quantum models and their practical efficiency to emulate neural networks, motivating alternative QML paradigms. 

\begin{figure*}[htbp]
\vspace{-.2cm}
    \centering
	\begin{subfigure}[t]{0.3\textwidth}
	      		\begin{tikzpicture}[
	      		node distance=3.8cm,
	      		box/.style={
	      			draw, rounded corners, thick, align=center, inner sep=3pt, minimum width=5.2cm
	      		}]
	      		%
	      		\node[scale=.8, box, label=above: {\textbf{Standard BVN}}] (standard) {
	      			\begin{tabular}{c}
	      				$\displaystyle f(x)=\sum_{y} \widehat f(y)\,\chi_y(x)$\\
	      				$\displaystyle \chi_y(x)=(-1)^{y \odot x}$
	      			\end{tabular}
	      		};
	      		%
	      		\node[above=1.cm of standard.center]
	      		{\includegraphics[scale=.75]{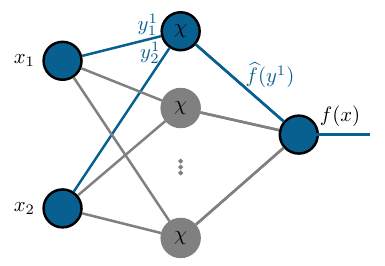}};
	      	\end{tikzpicture}
	      	\caption{Standard BVN}
	      	\label{fig:bv_model}
	\end{subfigure}
	\hspace{-.5cm}
	\begin{subfigure}[t]{0.3\textwidth}
			\begin{tikzpicture}[
			node distance=3.8cm,
			box/.style={
				draw, rounded corners, thick, align=center, inner sep=3pt, minimum width=5.2cm
			}]
			\node[scale=.8, box, label=above: {\textbf{Generalised BVN}}] (generalised) {
				\begin{tabular}{c}
					$\displaystyle f(x)\approx\sum_{y,z,t} \widehat f(yzt)\,\chi_{yzt}(x)$\\
					$\displaystyle \chi_{yzt}(x)=\text{more expressive feature}$
				\end{tabular}
			};
			%
			\node[above=1.cm of generalised.center]
			{\includegraphics[scale=.72]{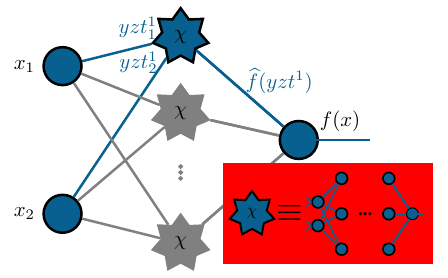}};
		\end{tikzpicture}
		\caption{ Generalised BVN}
		\label{fig:iqnn}
	\end{subfigure}
	\begin{subfigure}[t]{0.41\textwidth}
	 \includegraphics[width=1.\linewidth, trim={0cm .01cm 0cm -1.cm}, clip]{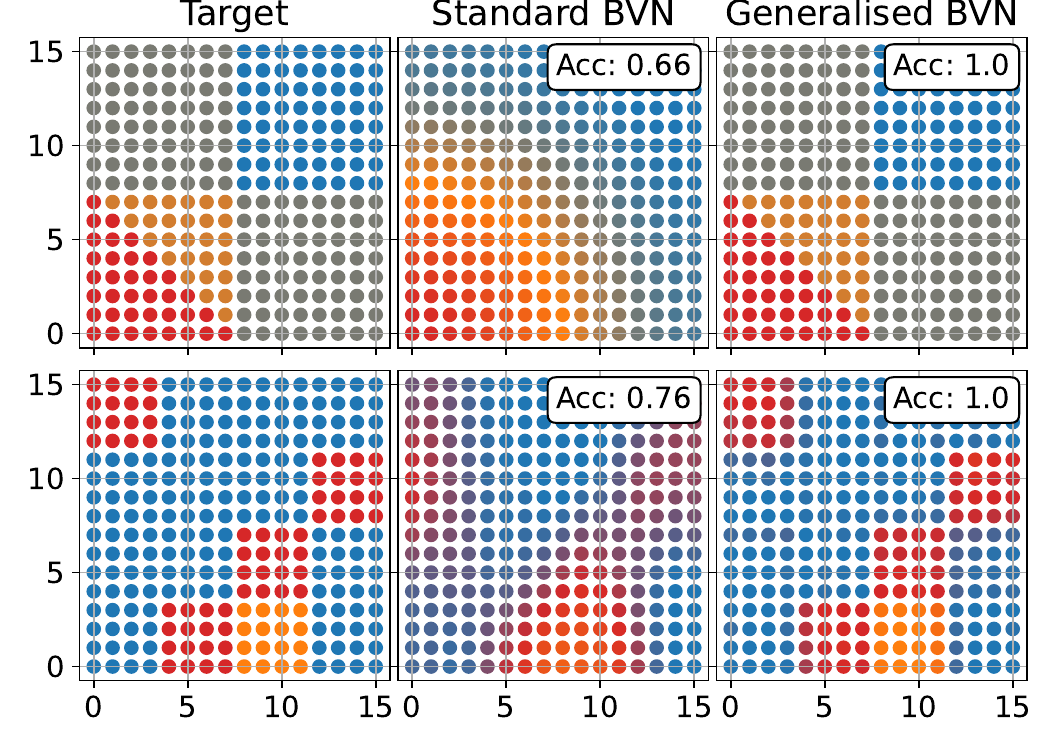}
     \vspace{-.5cm}
	 \caption{2D classification with BVNs}
	 \label{fig:2d_expressive_repr}
	\end{subfigure}

    \caption{
    \textbf{The proposed standard (\ref{fig:bv_model}) and generalised (\ref{fig:iqnn}) Bernstein--Vazirani Networks (BVNs).} 
    In the standard setting, binary inputs $x_i$ are linearly combined in a hidden layer with optimised binary weights $y^j$, passed through a parity activation $\chi$, and linearly combined at the output layer using Fourier coefficients $\widehat f(y^j)$.
    In the generalised setting, inputs pass through expressive optimised sub-models $\chi$ (cartooned in the red box) before the output combination, enabling richer representations.
     (\ref{fig:2d_expressive_repr}): Illustratively, we classify 2D points where the target function (left) is a sum of expressive sub-models $\chi$ of the generalised basis. 
    The generalised model fits the boundaries exactly, while the standard model has not yet converged. 
    ``Acc.'' denotes accuracy. 
    See Secs.~\ref{sec:method} and \ref{sec:exterimental_results} for more details. 
    \vspace{-2mm}
    }
    \label{fig:bv_and_iqnn_models}
\end{figure*}

From this perspective, it is instructive to revisit quantum algorithms that originally demonstrated provable advantages compared to classical counterparts, such as Shor \cite{shor1999polynomial}, Deutsch--Jozsa \cite{deutsch1992rapid} or Bernstein--Vazirani \cite{bernstein1993quantum}, which use quantum Fourier sampling to extract hidden structure via \emph{interference}. 
Wakeham and Schuld \cite{wakeham2024inference} identify a common blueprint underlying such algorithms that can be repurposed for inference tasks: \emph{place labelled information in superposition via an oracle, interfere in the Fourier space, and measure.} 
In this paradigm, learning is achieved via interference by sampling globally informative Fourier characters rather than via gradient-based optimisation, avoiding issues such as local minima and barren plateaus. 
This reframes learning alternatively as a problem of frequency discovery rather than parameter optimisation. 
Wakeham and Schuld, however, left unresolved the ``leakage'' problem that occurs when no single frequency fits the data perfectly, preventing them from validating the idea practically. 
Building on this blueprint, we introduce \emph{Bernstein--Vazirani Networks (BVNs)}, which \emph{repurpose} the Bernstein--Vazirani (BV) algorithm \cite{bernstein1993quantum} for learning tasks. 
The BV algorithm interferes labelled training examples placed in superposition to reveal a set of basis functions that explain the training data. 
We further generalise the standard BV framework---defining \emph{generalised BVNs}---by allowing interference in alternative problem-adapted bases, enabling to sample more diverse and expressive features while reducing the required measurement budget. 
Importantly, quantum circuits in BVNs require no variational parameter optimisation with respect to a loss function and rely solely on labelled data to infer the underlying target function that explains the observations. 
The expressivity of our models stems from the exponentially large Hilbert space and the (over)completeness of the interference basis that enable universal function approximation.
Intuitively, generalised BVNs can be viewed as expressive parametrised quantum architectures, \textit{in which all network parameter configurations are placed in superposition.} 
Interference then highlights the configurations that best align with the structure induced by the labelled data. 
Rather than identifying a single optimal configuration, the algorithm samples multiple high-quality configurations, which are aggregated to approximate the target function, addressing the leakage problem. 
Each sampled configuration corresponds to a sub-network within a larger implicit ensemble, evaluated at the cost of a single sub-network instance and optimised through interference rather than gradient descent. 
An intuitive illustration of the proposed BVNs is shown in \Cref{fig:bv_and_iqnn_models} with a hand-crafted 2D point classification result, where the target function is a sum of basis functions unique to the generalised basis. 
After only 100 shots, the generalised BVN fits the decision boundaries precisely, whereas the not-yet-converged standard BVN smooths them. 

In summary, our key technical contributions are as follows:
\begin{itemize}[leftmargin=*,itemsep=0.1em,topsep=0.1em]
    \item \textbf{Bernstein--Vazirani Networks (BVNs)}, a new class of QML models that leverage quantum interference to discover dominant basis functions for efficient function approximation in machine learning tasks; 
    \item \textbf{Generalised BVNs} that extend the standard framework with more expressive basis functions, improving data fitting under the same measurement budget. 
    \item \textbf{Coefficient reconstruction via ridge regression} to address the leakage problem and recover the best approximation of the target function in the sampled basis. 
\end{itemize}

Empirically, we validate our BVNs on several synthetic and real-world classification tasks  (Iris \cite{fisher1936use} and Penguins \cite{horst2022palmer} datasets), as well as 2D image-fitting tasks. 
Classification results demonstrate the ability of BVNs to learn decision boundaries and perform on par with classical (MLP, SVM) and quantum models (Single-Qubit Classifier \cite{perez2020data}), while being far more sampling-efficient than PQC models. 
In line with this, our BVNs continue to perform accurately on image representation tasks, competing with classical models (MLP+RFF, SIREN \cite{benbarka2022seeing}) and outperforming the recent PQC models (QIREN \cite{zhao2024quantum} and QVF \cite{wang2025quantum}) in both accuracy and \mbox{sampling cost}. 

\section{Related Works}
\label{sec:related_works}

Quantum Machine Learning (QML) \cite{biamonte2017quantum,schuld2018supervised,cerezo2022challenges} is driven by the idea of combining the power of quantum computing with the success of neural networks.
We review key developments in the field and refer the reader to a recent survey for further details~\cite{meli2025quantum}.

\paragraph{Quantum Neural Network Designs.}
Schuld et al.~\cite{schuld2014quest}, by surveying early work on Quantum Neural Networks (QNNs), outlined three requirements for QNNs: 
(i) a fixed-length bitstring initial state; 
(ii) neural‑like connections and update rules; and 
(iii) quantum‑consistent evolution exploiting effects such as superposition, entanglement, and interference.
These requirements motivated several elementary QNN architectures, including quantum neurons with ancillary registers \cite{cao2017quantum} and various quantum perceptron models \cite{beer2020training,da2016quantum,kapoor2016quantum}.
While replacing classical components with quantum counterparts is a reasonable direction, a central challenge in QML has long been the tension between the nonlinear, dissipative dynamics of neural computation and the linear, unitary dynamics of quantum mechanics \cite{schuld2014quest}. 
Herman et al. \cite{herman2023expressivity} showed that one way to bypass this issue is to map inputs to the Boolean cube.
Yet, training their models still relies on a classical optimisation loop.

Recently, Wakeham and Schuld \cite{wakeham2024inference} explored an alternative approach to QML drawn from quantum algorithms such as Deutsch–Jozsa \cite{deutsch1992rapid} and Shor's algorithm \cite{shor1999polynomial}, rather than from classical neural networks.
They highlight a common blueprint shared by these algorithms: \emph{place labelled information in superposition via an oracle, interfere in Fourier space, and measure}.
The authors then leverage this \emph{interference strategy} for \emph{inference}, by regarding a learning task as a hidden-subgroup problem in which interference should reveal the so-called ``annihilator'' that generated the data, i.e., the true oracle that produced the training examples.
This perspective suggests a promising route toward genuinely quantum learning models by directly aligning inference with intrinsic quantum mechanisms.
Wakeham and Schuld \cite{wakeham2024inference} left open how to handle leakage, i.e., loss of information when no perfect annihilator exists.
We address this gap by drawing on the Bernstein--Vazirani algorithm \cite{bernstein1993quantum}, validating the approach on several learning tasks and showing how ``leaked information'' can be aggregated. 

\paragraph{Variational Quantum Neural Networks.}
Variational quantum neural networks implemented via Parametrised Quantum Circuits (PQCs) are the dominant paradigm in contemporary QML.
Schuld and Petruccione \cite{schuld2018supervised} formalised supervised learning with quantum computers, introducing the embedding of classical data into quantum states and optimising PQCs as trainable models.
Key developments include PQC architectures designed for specific expressivities and inductive biases.
Notably, data re-uploading classifiers showed that a single qubit with repeated data encoding suffices for binary classification \cite{perez2020data}, later formalised as implementing a Fourier series \cite{schuld2021effect}. 
Further work analysed how entanglement patterns affect expressivity \cite{sim2019expressibility} and introduced quantum convolutional neural networks inspired by classical ones \cite{cong2019quantum}.

A major challenge for PQCs is barren plateaus \cite{larocca2025barren}, where flat loss landscapes hinder optimisation. 
Only a few architectures, mostly group-equivariant models with built-in inductive biases \cite{schatzki2024theoretical,liu2025rotation,nguyen2024theory} or shallow-depth PQCs, admit theoretical training guarantees. 
Our approach diverges fundamentally from this paradigm by eliminating classical outer-loop optimisation and instead exploiting quantum interference to determine network configurations.

\paragraph{Applications.}
Representative supervised learning tasks for assessing QML models include classification and implicit representation. 
For classification, Salinas et al.~introduced the Single-Qubit Classifier (SQC) \cite{perez2020data}, which associates target classes with specific regions of the Bloch sphere and trains PQCs to cluster input states in the corresponding class representations. 
The SQC employs data re-uploading to approximate functions through their Fourier representations.
This Fourier-based perspective was later extended to Quantum Implicit Neural Representation (QIREN) by Zhao et al.~\cite{zhao2024quantum}, where PQCs learn a mapping from coordinates to signal values, such as those describing images or audio. 
More recently, Wang et al.~\cite{wang2025quantum} introduced Quantum Visual Fields (QVFs), a hybrid architecture in which a classical backbone extracts and normalises feature vectors that are subsequently processed by PQCs, improving the expressive capacity of the quantum model. 

\section{Motivation}
\label{sec:background}
Supervised learning \cite{nielsen2015neural,Goodfellow-et-al-2016} seeks to approximate an unknown function $f$ from training pairs $\{(x^i, y^i)\}_{i=1}^m$, where $f(x^i)=y^i$.
This is typically achieved by a parameterised model $\mathcal N$ such that $\mathcal N(x^i,\theta)\approx f(x^i)$ for some optimal parameter set $\theta$.
This network $\mathcal N$ is often organised into layers, with an output layer aggregating the results of previous layers. 
We can express this as
\begin{equation}
	\label{eq:qml_repr}
	\mathcal N(x) = \sum_{j=1}^k w_j \phi_j(x),
\end{equation}
where $\phi_j(x)$ denotes the output of the $j$-th neuron in the last hidden layer and $w_j$ are the corresponding output weights (both collectively parametrised by $\theta$). 
In essence, we are representing the target function $f$ in a learned basis formed by the functions $\{\phi_j\}_{j=1}^k$ produced before the output layer.

Typically, training requires designing an expressive architecture $\mathcal N$, performing multiple passes over the data, and optimising non-trivial loss functions.
The expressive architecture is constructed so that the hidden layers, or basis functions $\{\phi_j\}_{j=1}^k$, lift the data into a higher-dimensional space, which in practice increases the number of parameters.
For this reason, it is believed that QC, which can access high-dimensional Hilbert spaces, may offer advantages in ML, giving rise to QML \cite{biamonte2017quantum,schuld2014quest,cerezo2022challenges}.

An interesting question we would like to answer is how to design a learning paradigm that, given access to the training dataset, can infer within a high-dimensional basis $\{\phi_j\}$ a reasonably sized subset $\{\phi_j\}_{j\in S}$, $|S| \ll |\{\phi_j\}|$, that approximates $f$ via \Cref{eq:qml_repr} sufficiently well.

Our BVNs operate at the natural intersection of learning and quantum computing in the sense that, by exploiting superposition and interference, they sample indexed basis functions from an exponentially large basis with probability proportional to the amplitudes of their weights in the representation of the target function $f$ in that basis.
These properties, together with the fact that BVNs (\Cref{sec:method}) avoid several fundamental limitations of gradient-based PQC training, make them a highly appealing learning paradigm.

\section{(Generalised) Bernstein--Vazirani Networks}
\label{sec:method}

We introduce our interference-based models, beginning with a warm-up on the Bernstein--Vazirani algorithm.
Key symbols and notations used in this work are listed in \Cref{tab:symbols}.

\begin{table}[t]
	\centering
	\caption{List of symbols used in this paper.}
\resizebox{1\linewidth}{!}{%
	\begin{tabular}{c|l}
		\toprule
		$\Z_2^n$  & $n$-dimensional Boolean space, $ \Z_2^n = \{0, 1\}^n$ \\
		$\odot$, $\oplus$  &  inner-product and sum $\bmod 2$, respectively \\
		$\tilde f$, $f$  &  standard and generalised target function \\
		$\hat f (\cdot)$  &  Fourier coefficient of $f$ \\
		$H$  &  Hadamard transform \\
		$I$  &  interference operator \\
		$ U_f$  &  oracle unitary for function $f$\\
		$\eta$  &  expressive representation \\
		$\xi$, $\chi$  &  standard and generalised basis function \\
		$\bra{}$, $\ket{}$  &  bra and ket, complex row and column vectors\\
		\bottomrule
	\end{tabular}
    }
	\vspace{-5mm}
	\label{tab:symbols}
\end{table}

\subsection{Warm-up: The Bernstein--Vazirani Algorithm}
\label{sec:warm_up}
The Bernstein--Vazirani (BV) algorithm \cite{bernstein1993quantum} is a celebrated example of quantum advantage. 
It solves the following problem:
Given a Boolean function $\tilde f : \Z_2^n \to \{0,1\}$ that is \emph{promised} to satisfy $\tilde f(x)=s\odot x = s^\top x \bmod 2$ for an \emph{unknown} $s\in \Z_2^n$, the goal is to recover $s$ from observed pairs $(x,\tilde f(x))$.

The algorithm, illustrated in \cref{fig:bv}, has three steps:
(i) starting from $\ket{0}^{\otimes n}$, we apply $H^{\otimes n}$ to create a uniform superposition over all basis states, i.e., inputs;
(ii) we query the oracle $U_{\tilde f}$, which applies the phase $(-1)^{\tilde f(x)}$ to each basis state in the superposition;
(iii) we apply $H^{\otimes n}$ again to interfere with the phases.
The corresponding states are
\begin{align}
	\label{eq:bv1}
	\ket{\psi_1} &= \frac{1}{\sqrt{2^n}} \hspace{-.1cm}\sum_{x\in \Z_2^n} \ket{x},\\
	\label{eq:bv2}
	\ket{\psi_2} &= \frac{1}{\sqrt{2^n}} \hspace{-.1cm}\sum_{x\in \Z_2^n} (-1)^{\tilde f(x)} \ket{x},\\
	\label{eq:bv3}
	\ket{\psi_3}&=\sum_{y\in \Z_2^n}\hspace{-.1cm} \widehat f(y)\ket{y},
	\quad
	\widehat f(y)\hspace{-.1cm}
	= \hspace{-.1cm}\frac{1}{2^n}\hspace{-.1cm} \sum_{x\in \Z_2^n} \hspace{-.1cm} (-1)^{\tilde f(x)+y\odot x}.
\end{align}
Substituting the promise $\tilde f(x)=s\odot x$ into \cref{eq:bv3} gives
$\widehat f(y)=\frac{1}{2^n}\sum_{x\in \Z_2^n} (-1)^{(s\oplus y)\odot x}$, which equals $1$ if and only if $y=s$, hence $0$ else (see \cref{app:bv}).
In essence, thanks to interference, a single measurement of the system reveals $s$ out of $2^n$ solutions with probability $1$.

\begin{figure}[t]
	\centering
	\vspace{-.3cm}
	\includegraphics[trim={0cm .2cm 0cm 0cm}]{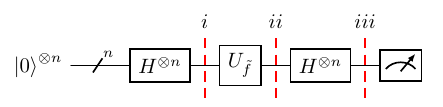}
	\caption{The Bernstein--Vazirani algorithm \cite{bernstein1993quantum}.  
		The algorithm mainly consists of three steps:  
		(i) initialise the system in a perfect superposition;  
		(ii) call the function oracle;  
		(iii) apply the Hadamard operator for interference.  
		If $\tilde f$ satisfies the promise $\tilde f(x) = s \odot x$, measuring reveals $s$ with certainty.\vspace{-3mm}
	}
	\label{fig:bv}
\end{figure}

\paragraph{Quantum Fourier Sampling.} 
The BV algorithm \cite{bernstein1993quantum} generalises to Fourier sampling over $\Z_2^n$:
If the promise $\tilde f(x)=s\odot x$ holds, we measure $s$ with probability $1$.
If not, consider the function space $\mathcal{F}:=\{f:\Z_2^n\to\R\}$ and the functions $\chi_y(x)=(-1)^{y\odot x}$, parameterised by $y\in\Z_2^n$, which encode the parity $y\odot x$. 
The set $\{\chi_y\}_{y\in\Z_2^n}$ forms an orthonormal basis of $\mathcal{F}$ with respect to the inner product $\braket{f,g}=\mathbb{E}_{x\in\Z_2^n}[f(x)g(x)]$, and we refer to $y\in\Z_2^n$ as a basis vector.
It follows that
\begin{equation}
	\widehat f(y) = \frac{1}{2^n} \sum_{x\in \Z_2^n} f(x)\chi_y(x)
\end{equation}
are precisely the Fourier coefficients of $f$ in this basis $\chi$.
\begin{theorem}[Fourier expansion on $\Z_2^n$, \protect{\cite{o2014analysis}}]
	\label{theo:fourier_bool}
	Any function $f: \mathbb{Z}_2^n \to \R$ admits a
	unique expansion as
	\begin{equation}
		\label{eq:fourier_boolean}
		f(x) = \sum_{y \in \Z_2^n} \widehat{f}(y) \, \chi_y(x).
	\end{equation}
	This is called the Fourier expansion of $f$, and the coefficients $\widehat{f}(y) \in \R$ are the Fourier coefficients of $f$ in the basis $\chi$. 
\end{theorem}
\begin{proof}
	\vspace{-.25cm}
	See \Cref{app:proof_fourier}.
\end{proof}

As implied by \Cref{theo:fourier_bool} and \Cref{eq:bv1,eq:bv2,eq:bv3}, running the BV algorithm for $f(x)=(-1)^{\tilde f(x)}$ or any other $f$ yields basis vectors $y \in \Z_2^n$ and corresponding functions $\chi_y \in \mathcal{F}$ that approximate $f$.
Moreover, each $y$ is sampled with probability equal to the squared of its Fourier coefficient $|\widehat f(y)|^2$, so components with large coefficients appear first.
This operation is known as Fourier sampling \cite{bernstein1993quantum}.
An alternative linear‑algebraic derivation of BVNs is given in \Cref{app:linear_algebra}.

\paragraph{A Neural View on the Bernstein--Vazirani Model.}
We can view the BV algorithm as a two-layer network, which we call a \emph{Bernstein--Vazirani Network (BVN)} and depict in \cref{fig:bv_model}. 
Hidden units compute parity activations $\chi_y(x)=(-1)^{y\odot x}$ on binary inputs $x_i$ weighted by $y_i$, and the output aggregates them with coefficients $\widehat f(y)$. 
If $f$ satisfies the BV promise, a single hidden unit suffices; otherwise, multiple units contribute according to $|\widehat f(y)|^2$. 
Importantly, this network ``growth'' comes purely from superposition and measurement, without additional quantum resources.

\subsection{Generalised Bernstein--Vazirani Networks}
\label{sec:iqnn}

The standard BV algorithm learns within a restricted hypothesis class: 
the basis elements $\chi_y$ are \emph{linear parity functions} in $\mathbb{Z}_2^n$, so the BV promise effectively assumes that $f$ is linear. 
Most practical learning functions will fail this assumption, possessing dense rather than sparse Fourier spectra and requiring many samples for the approximation.
Our goal in generalising the BV algorithm is to design task‑adaptive basis functions (i.e., the promise) with inductive biases that better align with the target function $f$. 
We achieve this along two complementary axes:

\begin{itemize}[leftmargin=*,itemsep=0.1em,topsep=0.1em]
	\item \textbf{Interference Operator.} 
	In standard BVNs, the Hadamard operator induces low-expressivity parity activations. 
	Replacing this alone directly implies alternative features.
	However, any alternative interference operator, being unitary, still primarily probes linear structure over $\mathbb{Z}_2^n$.
	\item \textbf{Expressive Input Representations.} 
	To capture non-linearity, we enlarge the quantum system and construct non-linear representations of the input $x$ in an auxiliary register prior to interference, making them potentially parameter-dependent to model implicit biases.
\end{itemize}

\textbf{Notation.}
We will call a unitary $I$ an \emph{interference operator} if its action on computational basis states $x \in \Z_2^n$ induces a family of functions $\{\xi_y\}_{y\in \Z_2^n}$ via
\begin{equation}
	I\ket{x} = \sum_{y \in \mathbb{Z}_2^n} \xi_y(x) \ket{y}.
\end{equation}

\paragraph{The Method.}

\begin{figure}[t]
	\centering
	\vspace{-.3cm}
	\includegraphics[trim={.2cm .3cm 0cm 0cm}, clip]{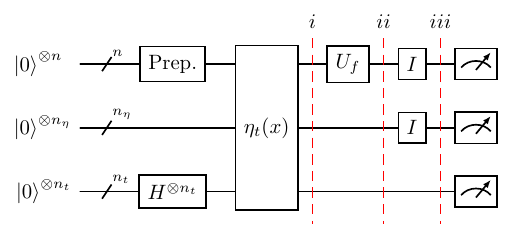}
	\caption{Our generalised Bernstein--Vazirani algorithm.
		We use three quantum registers: An input register, a register that computes a weighted representation of the inputs, and a weight register controlling this representation. 
		The algorithm has three stages:
		(i) Prepare a uniform superposition of all inputs (Prep.) and weight configurations, and compute the representations;
		(ii) Query the oracle encoding the ground‑truth labels across all samples in the superposition;
		(iii) Apply an interference operator so the labels interfere with the basis functions generated by the representation.
		\vspace{-3mm}
	}
	\label{fig:iqnn_circuit}
\end{figure}

We now propose a principled algorithm for
\emph{QML by interference}, inducing our \emph{generalised BVNs}.
To obtain expressive basis functions, we allocate two registers in addition to the input register:
One register that encodes a pre-processed \emph{representation} $\eta_t(x)\in \Z_2^{n_\eta}$ of the inputs $x \in \Z_2^{n}$, and one register that stores the \emph{weights or parameters} $t\in \Z_2^{n_t}$ associated with this representation.

Our generalised construction, as illustrated in \cref{fig:iqnn_circuit}, follows the three main steps of the Bernstein--Vazirani algorithm.
The inputs of the algorithm are training pairs $(x, f(x))_{i=1}^m$; a function oracle $U_f$ to label the inputs; and an interference operator $I$.
The main steps are as follows:

\begin{itemize}[leftmargin=*,itemsep=0.1em,topsep=0.1em]
	\item The initial step puts all training samples and parameters of the expressive representation in uniform superposition, then runs the representation circuit:
	\begin{align}
		\ket{\psi_1^{\prime\prime}} &= \frac{1}{\sqrt{m}} \sum_x \ket{x}\ket{0}\ket{0},\label{eq:prep1}\\
		\ket{\psi_1^\prime} &= \frac{1}{\sqrt{m2^{n_t}}} \sum_t \sum_x \ket{x}\ket{0}\ket{t},\text{ and }\label{eq:prep2}\\
		\ket{\psi_1} &= \frac{1}{\sqrt{m2^{n_t}}} \sum_t \sum_x \ket{x}\ket{\eta_t(x)}\ket{t}\label{eq:prep3}.
	\end{align}
	
	\item The second step queries the oracle to mark the ground-truth labels $f(x)$ as amplitudes on the inputs:
	\begin{equation}
		\label{eq:oracle_query}
		\ket{\psi_2} = \frac{1}{\sqrt{m2^{n_t}}} \sum_t \sum_x f(x) \ket{x}\ket{\eta_t(x)}\ket{t}.
	\end{equation}
	
	\item The third step creates, from the input and representation registers, interference of the function values: 
	\begin{equation}
		\label{eq:gbv_state}
		\ket{\psi_3} = \frac{1}{\sqrt{2^{n_t}}}\sum_{yzt} \widehat f(yzt)\ket{y}\ket{z}\ket{t},
	\end{equation}
	where
	\begin{align}
		\label{eq:coefs}
		\widehat f(yzt) &= \frac{1}{m}\sum_x f(x) \chi_{yzt}(x),\\
		\chi_{yzt}(x) &= \sqrt{m}\xi_y(x) \xi_z(\eta_t(x)).
	\end{align}
\end{itemize}

Sampling \Cref{eq:gbv_state} yields, with probability $|\widehat{f}(yzt)/\sqrt{2^{n_t}}|^2$, the dominant basis vectors 
$y \in \Z_2^{n}$, $z \in \Z_2^{n_\eta}$, $t \in \Z_2^{n_t}$, and thus the corresponding basis functions $\chi_{yzt}$.
Since the extended basis may be non‑orthonormal due to its over‑completeness relative to the effective number of training samples $m$, we must reconstruct each sampled coefficient $\widehat{f}(yzt)$, which is done efficiently as described next.

\paragraph{Reconstructing the Coefficients.}  
Let $X \in \R^{m \times k}$ denote the matrix of $k$ basis evaluations on the training set, $X_{ij} = \chi_j(x_i)$.  
Let $F \in \R^m$ be the vector of target values on the training set, and $\widehat{F} \in \R^k$ the vector of Fourier coefficients.  
We aim to recover $F$ from the basis $X$, or the best approximation of $f$ in $\chi$.  
In fact, the coefficients $\widehat{F}$ need to be scaled properly, as the sampling has learned the dominant correlations but not the overall scaling.
A standard approach to recover the coefficients is to solve the ridge regression problem $\min_{\widehat{F}} \Norm{F - X \widehat{F}}_2^2 + \lambda\Norm{\widehat{F}}_2^2$ via the Gram matrix:
\begin{equation}
	\label{eq:gram}
	G = X^\top X, \quad \widehat{F} \gets (G + \lambda I)^{-1} X^\top F,
\end{equation}
for a regularisation factor $\lambda \in \R_+$.
In practice, this reconstruction is comparable to a single training epoch in classical models, as it involves a single pass over the training set. 

\paragraph{Making New Predictions.}  
A new prediction on a data point $x$ is made by evaluating
\begin{equation}
	\label{eq:reconstruction}
	f(x) = \sum_{yzt} \widehat {f}(yzt) \chi_{yzt}(x).
\end{equation}
Our construction generalises the standard Bernstein--Vazirani algorithm in the sense that we sample BV basis functions $\xi_y(x)$ if $\xi_z(\eta_t(x))$ is a constant for all $x, z$, and $t$.
A summarising pseudocode of our generalised BVNs is provided in \Cref{alg:iqnn_blueprint}.

\begin{algorithm}[tb]
	\caption{Generalised Berstein--Vazirani Networks
		\phantom{Algorithm 1 }(Quantum Machine Learning by Interference)}
	\label{alg:iqnn_blueprint}
	\begin{algorithmic}
		\STATE {\bfseries Input:} Training pairs $(x, f(x))_{i=1}^m$; Function oracle $U_f$; \\ 
		Interference operator $I$; Number of shots \texttt{nshots}.
		\WHILE{\texttt{nshots} not reached}
		\STATE Prepare superposition of inputs $x$ via \Cref{eq:prep1}.
		
		\STATE Prepare superposition parameters $t$ via \Cref{eq:prep2}.
		
		\STATE Compute representation $\eta_t(x)$ via \Cref{eq:prep3}. 
		
		\STATE Query function oracle $U_f$ via \Cref{eq:oracle_query}. 
		
		\STATE Interfere the $x$ and $\eta$ registers via \Cref{eq:gbv_state}. 
		
		\STATE Measure basis states $ \{ \ket{y}\ket{z}\ket{t} \}$.
		\ENDWHILE
		
		\FOR{Parameter configuration $(y,z,t) \in \{\ket{y}\ket{z}\ket{t}\}$}
		\STATE Evaluate
		$\chi_{yzt}(\cdot) = \sqrt{m}\xi_y(\cdot) \xi_z(\eta_t(\cdot))$ on training set.
		\ENDFOR
		\STATE Compute Fourier coefficient $\widehat f(yzt)$ via \cref{eq:gram}.
		
		\STATE {\bfseries Return:} {Fourier pairs $(\widehat f(yzt), \chi_{yzt})$}.
	\end{algorithmic}
\end{algorithm}

\paragraph{The Generalised Bernstein--Vazirani Network.} 
The generalised BVN architecture is shown in \cref{fig:iqnn}. 
Unlike standard BVNs, the hidden units are no longer simple parity functions but expressive sub-models. 
Inputs $x_i$ pass through a hidden layer of optimised components $\chi_{yzt}$, whose predictions are aggregated at the output and weighted by the interference amplitudes $\widehat f(yzt)$. 
The architecture behaves similarly to an ensemble, but the ensemble is generated implicitly through interference rather than by training multiple models. 
Crucially, \emph{this rich ensemble incurs no additional quantum resource cost: only a single model is physically implemented, while the remaining effective models emerge purely through superposition and measurement.}

\subsection{The Expressive Representation of Inputs}
\label{sec:expressive_repr}
We now elaborate on our \emph{expressive representation} $\eta$ of the inputs.
This representation is central to obtaining smooth, regularised models that generalise well.

In principle \citep[Section 1.4.1]{nielsen2010quantum}, any Boolean function on binary inputs can be computed coherently: there exists a unitary $A$ (constructed from Toffoli and CNOT gates; and ancilla qubits) such that
\begin{equation}
	\label{eq:representation}
	A\ket{x}\ket{0}\ket{t} = \ket{x}\ket{\eta_t(x)}\ket{t}.
\end{equation}
That is, as an expressive representation, we could mirror standard classical network architectures—whose effectiveness is well established—to construct $\eta$ in a reversible quantum form. 
\Cref{app:layered_repr} discusses a quantum implementation of a classical 2-layer MLP with step activation functions and fixed biases.
In practice, however, implementing $A$ to mirror classical models requires ancilla qubits and uncomputation, and induces nontrivial circuit depths.

Constrained by the available qubits, we introduce a structured, resource‑efficient rectangle representation $\eta_{\text{rect}}$. 
For instance, since the standard BV basis is already complete, our rectangle representations activate parameter‑dependent hyper‑rectangles of the input domain.
This has the effect, after interference, of modifying the BV basis locally on different rectangles.
Implementation details and ablations are discussed in \cref{app:rectangle_repr}.

\subsection{The Interference Operator}
\label{sec:interference_ops}
We lastly discuss three candidate interference operators: Hadamard, Fourier, and Chebyshev. 
The interference operator determines the hidden activation and model expressivity, e.g., Hadamard produces step-like patterns, while Fourier and Chebyshev yield smoother variations (see \cref{app:interference_ops}).

\paragraph{The Hadamard Transform.}
The Hadamard transform $H$ \cite{nielsen2010quantum} is the most natural interference operator on $\mathbb{Z}_2^n$.
It acts on a basis state $\ket{x}$ as
\begin{equation}
	H\ket{x} = \frac{1}{\sqrt{2^n}} \sum_{y\in \mathbb{Z}_2^n} (-1)^{x \odot y} \ket{y},
\end{equation}
where $x$ and $y$ are binary vectors.  
The phase factor $(-1)^{x\odot y}$ encodes the parity of the binary inner product.
The Hadamard transform is the Fourier transform on $\mathbb{Z}_2^n$.

\paragraph{The Fourier Transform.}
A more general Fourier transform $F$ \cite{nielsen2010quantum} arises when we embed $\mathbb{Z}_2^n$ into the complex vector space indexed by integers in $\mathbb{Z}_{2^n}$.
In this viewpoint, the $x$ is interpreted as an integer in $\{0,\dots,2^n-1\}$.
The transform acts on basis states as
\begin{equation}
	F\ket{x}
	= \frac{1}{\sqrt{2^n}}
	\sum_{y\in \mathbb{Z}_2^n}
	e^{i 2\pi xy / 2^n} \ket{y},
\end{equation}
where $xy$ denotes standard integer multiplication.

\paragraph{The Chebyshev Transform.}
The Chebyshev transform $T$ \cite{williams2023quantum} constructs Chebyshev polynomials on discretised Chebyshev nodes.
For continuous inputs $x\in[-1,1]$, Chebyshev polynomials are defined as $\chi_y(x)=\cos(y\arccos(x))$.
In the discrete setting, we view $x\in\mathbb{Z}_2^n$ as an integer and map it to a Chebyshev node $x^{\mathrm{ch}} := \cos \left( \pi (2x + 1) /2^{n+1} \right)$.
Defining polynomials $T_y(x) := \cos \left( y \arccos(x^{\mathrm{ch}}) \right)$, the transform acts as
\begin{equation}
	T\ket{x}
	= \frac{1}{2^{\frac{n}{2}}}\, T_0(x)\ket{0}
	+ \frac{1}{2^{\frac{n-1}{2}}}
	\sum_{y\in \mathbb{Z}_2^n \setminus \{0\}}
	T_y(x)\ket{y}.
\end{equation}

\section{Experimental Results}
\label{sec:exterimental_results}
We validate our BVNs on several representative supervised learning tasks of adequate size for classical simulation, from 
hand‑crafted sanity-check tasks \Cref{sec:2d_expressive_repr}), through 
synthetic 2D and real‑world 4D classification (\Cref{sec:2d_classification}), to implicit image representation (\Cref{sec:im_repr}). 
\Cref{sec:complexities} analyses complexities and compares qubits, parameters, shots and runtime costs, with extended discussions covering the sampling complexity in \Cref{app:complexities_and_comparision}.
\Cref{sec:ablations} lastly performs ablation studies on the proposed BVNs.
All quantum circuits, unless otherwise stated, are simulated noise‑free using Pennylane’s \texttt{lightning.qubit} \cite{bergholm2018pennylane} to investigate properties of the new paradigm.
We compute on a conventional computer (AMD Ryzen 9 5900X 12-Core Processor CPU with 128GB RAM).
The classification metric is the top-1 accuracy between ground-truths and rounded predictions.

\subsection{Fitting the Expressive Representation in 2D}
\label{sec:2d_expressive_repr}
In \Cref{fig:2d_expressive_repr} we demonstrate the efficiency of the generalised model over the standard one, constructing 2D classification tasks whose regions match our expressive basis functions, making the target directly representable in the generalised basis. 
With Chebyshev interference and only $100$ shots, the generalised BVN achieves $100\%$ accuracy, whereas the standard BVN reaches only ${\sim}70\%$. 

\subsection{Classification in 2D and 4D}
\label{sec:2d_classification}
We experiment with classification tasks.
First, we solve binary classification in 2D on a variety of synthetic datasets shown in \cref{fig:shapes}, including the canonical blobs, moons, and a circle, as well as several hand-made shapes with different boundaries, yielding ten datasets S1-S10.
Second, we classify the real-world Iris \cite{fisher1936use} ($150$ samples) and Penguins \cite{horst2022palmer} datasets ($333$ samples), each consisting of 4D vectors across three classes (1,2,3).  
Data are normalised to the qubit grid. 
We benchmark our BVNs with Hadamard and Chebyshev operators against a classical $3$-layer MLP baseline with 10 units per layer, a classical SVM model, and a $20$-layer universal quantum single-qubit classifier (SQC) \cite{perez2020data}. 
We vary the training set size ($25\%, 50\%, 75\%, 100\%$) and evaluate on the full datasets to assess how well the models \emph{generalise}, i.e., learn the classification boundaries.
\Cref{app:inspection} inspects sampled pairs $(\widehat f, \chi)$. 
Results report the mean performance over five runs per model and dataset.

\Cref{fig:2d_classification} presents the binary classification results in 2D. 
Our BVNs successfully learn the decision boundaries.  
With only $50\%$ of the training data, they reach nearly $100\%$ accuracy.  
In the low training split, we observe that increasing the regularisation parameter $\lambda$ improves generalisation.  
The MLP, the SVM and single-qubit classifier baselines struggle on datasets with high variation (S2, S6, S9), whereas the BVN bases are expressive enough to capture high variations.  
The MLP’s behaviour reflects its spectral bias, and we conjecture that the SQC is constrained by its depth.

\Cref{fig:4d_classification} shows the 4D results. 
Here, we observed that the small number of training samples (${<}500$) relative to the large Hilbert‑space dimension ($2^{16}$) leads to severe sampling leakage \cite{wakeham2024inference}. 
Assigning a dummy non‑zero fill‑label to inputs outside the training set mitigated this issue. 
We use $\lambda=0.1$, refer to \Cref{app:filled_comparison} for other values of $\lambda$ and to \Cref{sec:dummy_fill} for Ablations on fill fraction and value. 
Unfilled data lead to poor generalisation, whereas filled data, while degrading the standard BVN, substantially improve the generalised BVN, bringing it on par with baselines.
We attribute the SVM's slightly better performance over the MLP to the small size of the datasets.

\begin{figure*}[t]
	
	\centering
	\hfill
	\begin{subfigure}{.13\textwidth}
		\vspace{-.1cm}
		\includegraphics[width=1.15\linewidth]{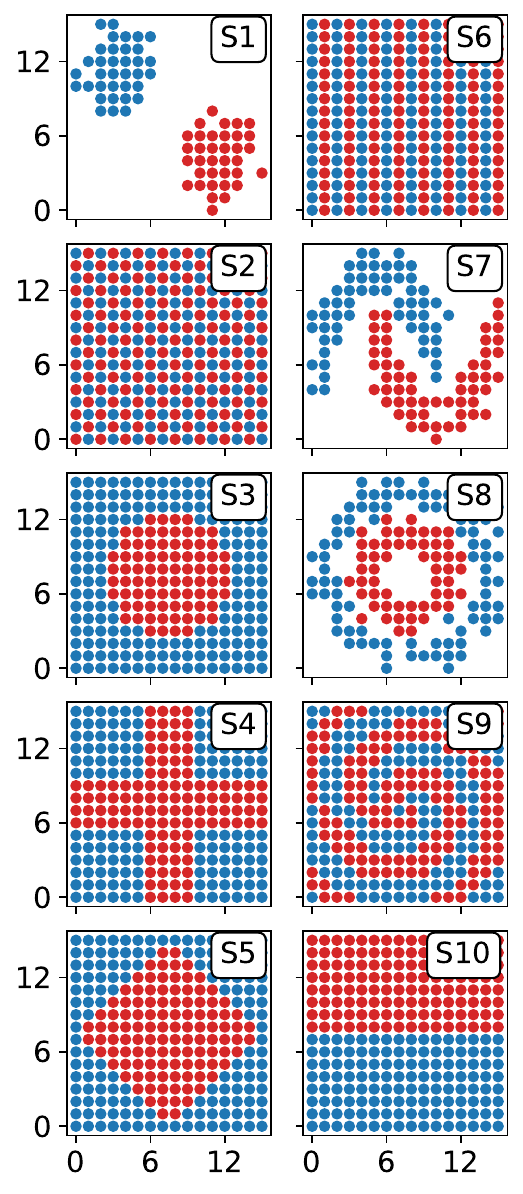}
		\caption{2D datasets}
		\label{fig:shapes}
	\end{subfigure}
	\hfill
	\begin{subfigure}{.8\textwidth}
		\vspace{-.2cm}
		\includegraphics[width=.95\linewidth, trim={0cm .26cm 0cm .2cm}, clip]{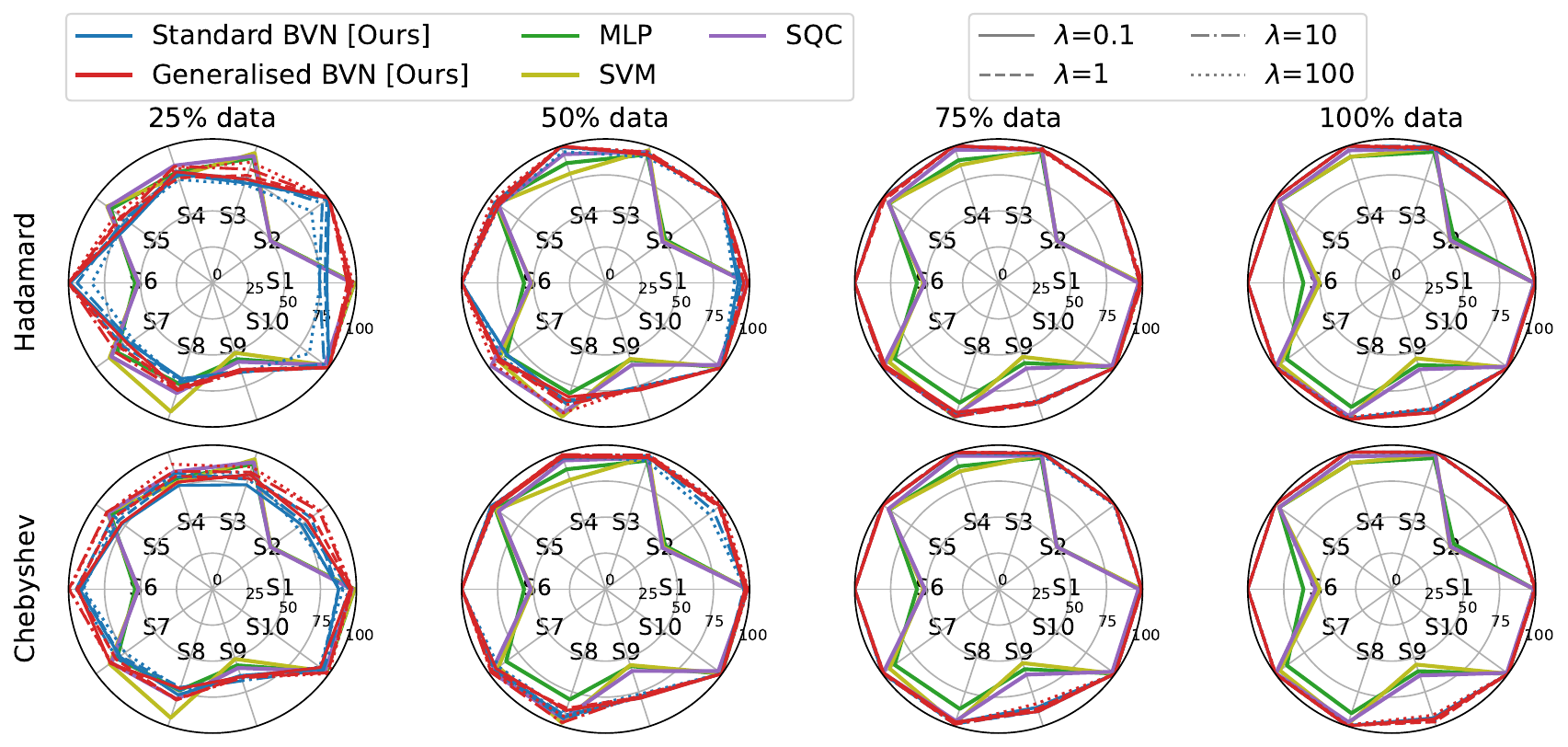}
		\caption{Classification results}
		\label{fig:2d_classification}
	\end{subfigure}
	\vspace{-2.mm}
	\caption{
		Binary classification results on synthetic 2D datasets.
		The top and bottom rows show BVN results with Hadamard and Chebyshev interference, with columns specifying data fractions for training. 
		Chebyshev outperforms Hadamard; the generalised BVN outperforms the standard BVN. 
		BVNs outperform MLP, SVM, and SQC, particularly on high-frequency shapes (S2, S6, and S9).\vspace{-3mm}
	}
	\label{fig:2d_classification_all}
\end{figure*}

\begin{figure}
	\centering
	\includegraphics[width=.93\linewidth, trim={0cm .3cm 0cm .1cm}, clip]{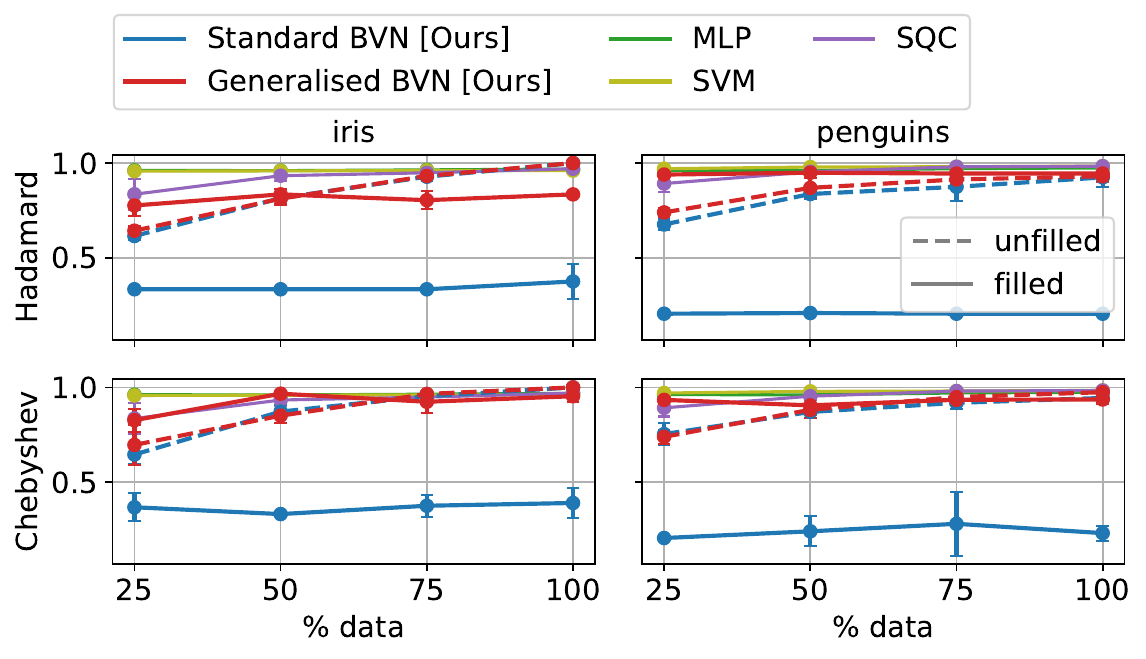}
	\caption{
		Real‑world classification results on Iris and Penguins datasets, with filled vs.~unfilled (original) comparison for BVNs.
		Filled data substantially improves the generalised BVN, allowing ${>}90\%$ accuracy with only $25\%$ data, comparable to the baselines.\vspace{-4.5mm}
	}
	\label{fig:4d_classification}
\end{figure}

\vspace{-1.5mm}
\subsection{Implicit Image Representation}
\label{sec:im_repr}
\vspace{-.1cm}
We run BVNs with $\lambda=0.1$ and Chebyshev on implicit image representations. 
We benchmark against the classical MLP+RFF and SIREN models \cite{benbarka2022seeing}, and the quantum QIREN \cite{zhao2024quantum} and QVF \cite{wang2025quantum} models. 
The models are trained on a coarse 64×64 image, evaluated on a finer 128×128 grid to assess interpolation, and we report PSNR and MSE.
We use 10,000 shots to demonstrate that BVNs can represent high‑frequency signals. 
Ablations in \cref{fig:shots_img_repr_sigmoid,fig:shots_img_repr_mask} show that even 1,000 shots yield good results visually.
\Cref{fig:image_fittind_results} presents our results. 
Again, the generalised BVN outperforms the standard one and stands competitively alongside the baselines.  
While it does not always yield the highest PSNR, it is visually more regular and coherent.

\section{Computational Complexity}
\label{sec:complexities}
Our labelled dataset encoding in BVNs has a cost of $O(mn)$ for $m$ training samples and $n$ qubits.  
The rectangle representations with input and parameter resolutions $n$ and $n_t$ yield $O(nn_t)$ gates.  
The Hadamard \cite{nielsen2010quantum} or Chebyshev \cite{williams2023quantum} operator costs $O(n)$ or $O(n+n^2)$ gates.  
Coefficient reconstruction scales as $O(mk(c+k+1) + k^3)$ for $k$ sampled states. 
We provide a full discussion in \Cref{app:complexities_and_comparision}.

\begin{figure*}[t]
    \centering
    \vspace{-3mm}
    \includegraphics[width=1\linewidth, trim={0cm 1.cm 0cm 0cm}, clip]{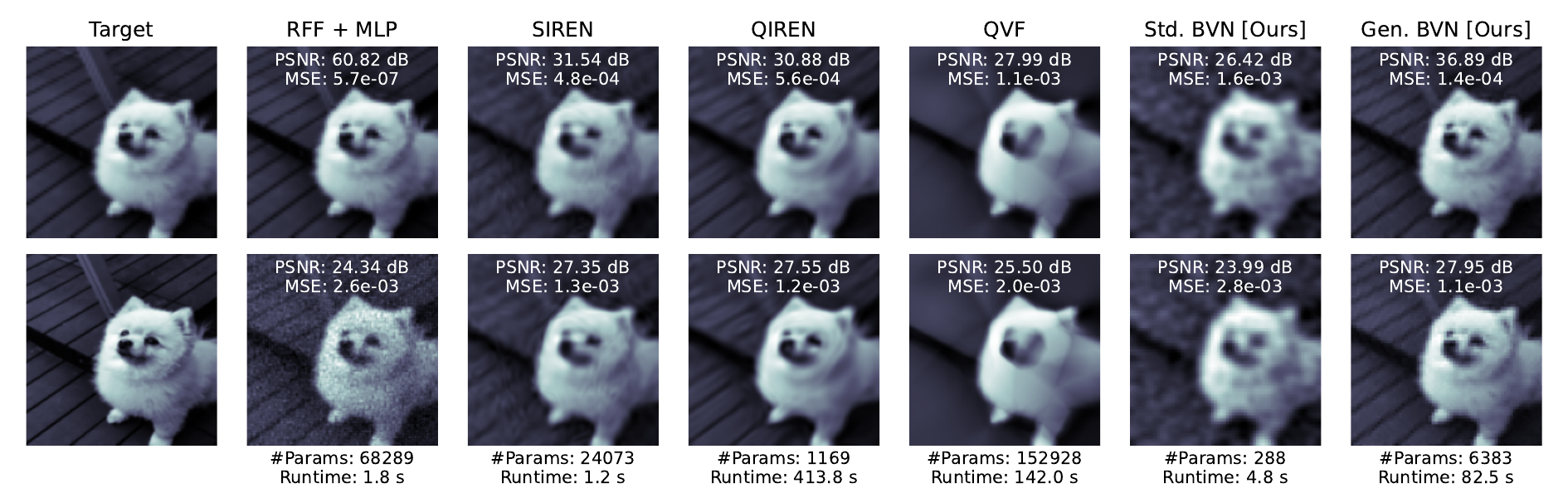}
    \caption{
    Implicit image representation.
    The top row shows results on the coarse (64×64) training image, while the bottom row shows evaluations on a finer (128×128) test image.
    The generalised BVN outperforms the standard BVN and is competitive with the benchmarks.
    }
    \label{fig:image_fittind_results}
\end{figure*}

\begin{table*}[t]
\centering
\caption{
Resource comparison across all models (quantum methods in blue).
We report the expected shot count for parameter‑shift training of PQCs and the runtime using back-propagation.
Our BVNs are significantly more efficient in terms of shots and runtime.
}
\label{tab:comparison}
\resizebox{0.8\linewidth}{!}{%
\begin{tabular}{l|c|c|c|c|c|c}
Model & \makecell{Num.~of\\qubits} & \makecell{Params \\ (PQC)} & \makecell{Params \\ (class.)} & \makecell{Epochs \\$\equiv$ Basis states for BVNs} & \makecell{Shots \\(parameter-shift)} & \makecell{Runtime \\ (back-prob.)} \\

\hline
\multicolumn{7}{c}{\textbf{Classification (2D), \Cref{fig:2d_classification_all}}} \\
\hline
MLP       & -- & -- & 151 & 30 & -- & 2.70s \\
SVM       & -- & -- & 253 & -- & -- & 0.001s \\
\rowcolor{quantumrow}SQC       & \textbf{1} & 120 & -- & 30 & $120\cdot30(2c)$ & 37.5s \\
\rowcolor{quantumrow}Std.\ BVN & \textit{\textbf{8(+1)}} & -- & -- & 45 & \textbf{100} & \textbf{0.027s} \\
\rowcolor{quantumrow}Gen.\ BVN & 14(+1) & -- & -- & 95 & \textbf{100} & \textit{\textbf{0.039s}} \\

\hline
\multicolumn{7}{c}{\textbf{Classification (4D), \Cref{fig:4d_classification}}} \\
\hline
MLP       & -- & -- & 171 & 100 & -- & 18.34s \\
SVM       & -- & -- & 178 & -- & -- & 0.0015s \\
\rowcolor{quantumrow}SQC       & \textbf{1} & 240 & -- & 30 & $240\cdot30(2c)$ & 819.35s \\
\rowcolor{quantumrow}Std.\ BVN & \textit{\textbf{16}(+1)} & --  & -- & 95 & \textbf{100} & \textbf{1.36s} \\
\rowcolor{quantumrow}Gen.\ BVN & 22(+1)  & -- & -- & 100 & \textbf{100} & \textit{\textbf{4.65s}} \\
\hline

\multicolumn{7}{c}{\textbf{Image Fitting, \Cref{fig:image_fittind_results}}} \\
\hline
RFF+MLP   & -- & -- & 68,289 & 100 & -- & 1.8s \\
SIREN     & -- & -- & 24,073 & 100 & -- & 1.2s \\
\rowcolor{quantumrow}QIREN     & \textit{\textbf{6}} & 270 & 899 & 2,000 & $270\cdot2000(2c)$ & 413.8s \\
\rowcolor{quantumrow}QVF       & \textbf{4} & 80 & 152,848 & 2,000 & $80\cdot2000(2c)$ & 142.0s \\
\rowcolor{quantumrow}Std.\ BVN & 12(+1) & -- & -- & 288 & \textbf{10,000} & \textbf{4.8s} \\
\rowcolor{quantumrow}Gen.\ BVN & 24(+1) & -- & -- & 6,383 & \textbf{10,000} & \textit{\textbf{82.5s}}
\end{tabular}
}
\vspace{-.2cm}
\end{table*}

\emph{
The total data-loading cost of BVNs, given by $O(mn)\times\texttt{nshots}$, scales significantly more favorably than that of PQC models, which is $O(mn)\times(2c\cdot p\cdot e)$, where $2c$ denotes the number of shots used for each partial derivative under the parameter-shift rule, $p$ is the number of trainable parameters, and $e$ is the number of training epochs.
}
\Cref{tab:comparison} provides a systematic resource comparison between BVNs and benchmarks, including the estimated shot count of parameter-shift training and true simulation runtime using back-propagation.
BVNs are markedly efficient: our 8-22-qubit BVNs are 2-3 orders of magnitude faster than the SQC or the multi‑qubit QIREN and QVF.

\paragraph{Approximation Error.}
Lastly, \Cref{app:sample_complexities} bounds the approximation error of BVNs.
Let  $\Omega\subset\mathbb Z_2^n$ be a sampled subset of basis states and $f_\Omega$ the truncated approximation of $f$ using basis states in $\Omega$.
The approximation error $\epsilon=\mathbb E_x[(f(x)-f_\Omega(x))^2]$ can be written as $\epsilon = \widehat F_{\bar\Omega}^{\top} G_{\bar\Omega} \widehat F_{\bar\Omega}$, with $\widehat F_{\bar\Omega} = (\widehat f(a))_{a \notin \Omega}$ denoting the residual coefficients and $G_{\bar\Omega}$ the Gram matrix from \Cref{eq:gram}.
Since $G_{\bar\Omega}$ is positive semidefinite, we have $\epsilon \le \lambda_{\max}(G_{\bar\Omega})\|\widehat F_{\bar\Omega}\|_2^2$, where $\lambda_{\max}(G_{\bar\Omega})$ is its largest eigenvalue.

For standard BVNs, we have $\|\widehat F_{\mathbb Z_2^n}\|_2=1$ and $G=I$, yielding $\epsilon = \|\widehat F_{\bar\Omega}\|_2^2 = 1-\|\widehat F_{\Omega}\|_2^2$.
Thus, the approximation error is determined entirely by the unsampled spectrum.
For generalised BVNs, the induced basis is generally overcomplete and non-orthogonal, yielding
$\epsilon \le \lambda_{\max}(G_{\bar\Omega})\|\widehat F_{\bar\Omega}\|_2^2$, showing a trade-off between the non-orthogonality of the extended basis and the expressive representation: 
Stronger correlations between basis functions increase the penalty factor $\lambda_{\max}(G_{\bar\Omega})$, while the richer representation $\eta_t$ can yield expressive representations of the target that substantially reduce the residual norm $\|\widehat F_{\bar\Omega}\|_2^2$.

\section{Ablations}
\label{sec:ablations}

\subsection{Dummy Filling}
\label{sec:dummy_fill}
We investigate the sensitivity of the generalised BVN to the dummy-fill value and fraction. 
We vary both and report results in \Cref{tab_dummy_filling_ablation} using tuples (value, fraction), with (0, 0) denoting unfilled data and (4, 1) the setting used in \Cref{fig:4d_classification}.

Benchmark methods are omitted as they fail on the filled data due to severe class imbalance: the $m<500$ informative samples are expanded to $2^{16}$ points, with the remaining points assigned a constant dummy label. 
In contrast, the generalised BVN remains performant across a wide range of fill values and fractions.

\begin{table}
	\caption{Influence of the dummy filling on the performance of BVNs. 
    Tuples (V, F) indicate dummy fill (value, fraction).} 
	\label{tab_dummy_filling_ablation}
    \centering
\resizebox{1.\linewidth}{!}{%
	\begin{tabular}{llcccccccc}
		\toprule
		& & \rotatebox{50}{(0, 0.0)} & \rotatebox{50}{(4, 0.25)} & \rotatebox{50}{(4, 0.5)} & \rotatebox{50}{(4, 0.75)} & \rotatebox{50}{(4, 1)} & \rotatebox{50}{(10, 1)} & \rotatebox{50}{(100, 1)} & \rotatebox{50}{(1000, 1)} \\
		\midrule
		\multirow[t]{4}{*}{Iris} & Stad. BVN & 0.84 & 0.78 & 0.65 & 0.39 & 0.33 & 0.33 & 0.33 & 0.33 \\
		& Gen. BVN & 0.77 & 0.76 & 0.74 & 0.99 & 0.88 & 0.99 & 0.93 & 0.85 \\
		\cline{1-10}
		\multirow[t]{4}{*}{Penguins} & Std. BVN & 0.60 & 0.38 & 0.40 & 0.23 & 0.20 & 0.20 & 0.20 & 0.20 \\
		& Gen. BVN & 0.59 & 0.86 & 0.85 & 0.87 & 0.96 & 0.93 & 0.95 & 0.96 \\
		\cline{1-10}
		\bottomrule
	\end{tabular}
    }
\vspace{-.5cm}
\end{table}

\vspace{-.2cm}
\subsection{Random Sampling}
\label{sec:random}
To isolate the value of quantum interference, we keep the coefficient reconstruction step from \Cref{eq:gram} unchanged and replace the quantum-selected basis functions with uniformly random basis functions of identical cardinality. 

\paragraph{Expected Coverage.} 
For an $n$-qubit input and an $(n_{\eta}+n_t)$-qubit extended register for expressive representation, the total measurement space dimension is $2^{n+n_{\eta}+n_t}$, but the relevant diversity for the input features is determined by the marginal distribution over the $n$-bit input register: $p(y) = \sum_{zt}p(y, zt)$.
For uniform sampling, $(y,z,t) \sim \mathcal{U}(2^{n+n_{\eta}+n_t})$, the probability of observing $y$ is $p(y) = 2^{n_{\eta}+n_t}/2^{n+n_{\eta}+n_t} = 1/2^{n}$, showing that the additional qubits do not improve the marginal probability of sampling a distinct input feature $y$.
Yet, the expected coverage, or number of unique $y$, after $N$ random draws, is
\begin{align}
	E &= \text{(number of y)} \cdot (1-\text{prob(y never occurs)}) \\
	&= 2^n(1-(1-1/2^n)^N) \approx 2^n(1-e^{-N/2^n}),
\end{align}
which shows that the number of shots for covering a constant fraction of the input space scales exponentially with $n$.

\paragraph{Empirical Evidence.}
In \Cref{fig:random_ablation}, we validate the above analysis by comparing the performance of the approximated function using basis functions identified by our BVNs and uniform random sampling.
We report average classification accuracy across the synthetic shapes in \Cref{fig:2d_classification_all}, the real-world filled Iris and Penguins datasets in \Cref{fig:4d_classification}, and PSNRs from the image fitting experiment in \Cref{fig:image_fittind_results}.

\begin{figure}
	\centering
	\includegraphics[width=.88\linewidth]{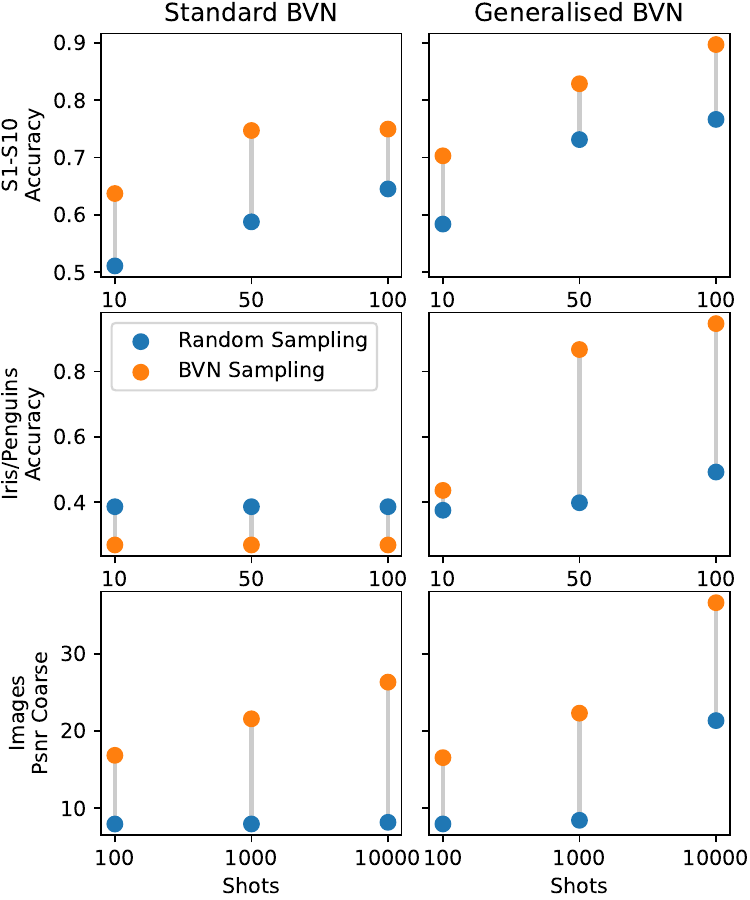}
	\caption{
		Ablation of interference-based sampling versus random uniform sampling. 
        BVN's performance indicates that interference effectively identifies informative basis functions.
\vspace{-.7cm}
	}
	\label{fig:random_ablation}
\end{figure}

BVNs consistently outperform random sampling.
As expected, the advantage is amplified as the search space grows: 
compare synthetic and real-world results using $N=100$ shots for $n=8/14$ and $n=16/22$ qubits in the standard/generalised BVN, respectively. 
For the image fitting experiment ($n=12$), random sampling benefits from a larger measurement budget at $10^4$ shots, reaching $~\sim$23 dB PSNR. 
Yet, both BVNs rapidly identify informative basis functions, particularly in the low-shot regimes ($N=100,1000$). 

\subsection{Robustness to Noise}
\label{sec:noise}
To evaluate robustness to hardware noise, we compare BVNs and SQC \cite{perez2020data} on the S8 shape of \Cref{fig:shapes} using a NISQ-inspired noise model \cite{abughanem2026early}:
Depolarizing errors of $5\times10^{-3}$ on 1-qubit gates, $5\times10^{-2}$ on 2-qubit gates, and a $2\%$ readout bit-flip error.
State preparation in BVNs is decomposed into elementary gates to expose this step to noise.
SQC is trained using backpropagation without noise and parameter-shift with noise. 
We report accuracy and runtime for both settings and compare BVNs against random sampling from \Cref{sec:random}. 
As shown in \Cref{tab:noise_ablation}, BVNs remain noise-robust and outperform random selection. 
SQC performs the best, but requires $\sim$3.2 hours of parameter-shift training, while BVNs take $<$15 seconds at competitive accuracy.

\begin{table}
	\centering
	\caption{Accuracy and runtime comparison in noise simulation. The numbers in brackets are in seconds. The experimental setting does not apply for classical methods (denoted by ``-'').}
	\label{tab:noise_ablation}
	\resizebox{1.\linewidth}{!}{%
		\begin{tabular}{llll}
			\toprule
			& Noise-free & Noise & Random \\
			\midrule
			\rowcolor{quantumrow}Std. BVN & 0.99 (0.005) & 0.86 (13.610) & 0.80 (0.005) \\
			\rowcolor{quantumrow}Gen. BVN & 0.98 (0.028) & 0.92 (13.525) & 0.85 (0.028) \\
			\rowcolor{quantumrow}SQC & 0.96 (112.066) & 0.99 (11513.526)	 & -\\
			MLP & 0.94 (3.801) & -& -\\
			SVM & 0.98 (0.010) & -& -\\
			\bottomrule
		\end{tabular}
	}	
	\vspace{-.4cm}
\end{table}

\section{Practical Recipe for BVNs}
\label{sec:recipe}
A central design question in BVNs is how to choose the interference operator, representation, and parameter resolution for a given application.
The key principle is to select them such that the induced basis aligns with the target function:
\begin{itemize}[leftmargin=*,itemsep=0.1em,topsep=0.1em]
    \item We replace the standard BV basis $\xi_y(x)=(-1)^{x\odot y}$ with $\chi_{yzt}(x)=\sqrt{m}\xi_y(x)\xi_z(\eta_t(x))$ to obtain potentially sparser representations. 
    For example, a function requiring two BV basis states in $\xi$ may become one-sparse in $\chi$ under a suitable representation.
    \item Beyond sparsification, generalised BVNs enable spectrum reshaping. 
    For example, a standard representation with coefficients $\sqrt{5/100}$ and $\sqrt{95/100}$ makes the smaller informative component unlikely to be sampled. 
    A suitable representation can redistribute this to make informative components more accessible.
\end{itemize}
Ultimately, these design choices should be guided by assumed structures in the data, like convolutional architectures for images or attention mechanisms for sequences.
A practical heuristic is to choose an interference operator expected to concentrate the target spectrum, a representation that encodes task-specific structures or invariances, and vary the parameter resolution to control the resulting expressivity.

\section{Conclusion}
\label{sec:conclusion}

We proposed an interference‑based alternative to variational QML by reinterpreting the Bernstein--Vazirani algorithm as a neural model, yielding Bernstein--Vazirani Networks (BVNs).
Our generalised BVNs incorporate inductive bias to improve fitting and reduce measurement cost. 
Experimentally, the generalised BVNs achieved nearly $100\%$ classification accuracy using only $50\%$ of the training data, and reached close to $40$ dB PSNR on image‑fitting, consistently surpassing the standard model and often outperforming quantum and classical baselines. 
Our work advances QML by introducing a conceptually new framework, which, we believe, will lead to a plethora of improvements.

\textbf{Limitations and Future Work.} 
While our BVNs make a first step in a fundamentally new QML direction, many open questions remain. 
For instance, while the generalised BVN performs well on synthetic data, scaling to real‑world data (e.g., Iris, Penguins) features spectral leakage. 
Beyond our ad-hoc dummy‑fill, basis engineering to align the parameter manifold with the data could help. 
Also, rectangle representations, while expressive, introduce artefacts and may hurt interpolation. 
Future work would explore smoother alternatives. 
We detail these directions in \Cref{app:future_work}.
%


\section*{Acknowledgements} 
This work was supported by the Deutsche Forschungsgemeinschaft (DFG, German Research Foundation), project number 534951134. 
T. Birdal was supported by a UKRI Future Leaders Fellowship [grant number MR/Y018818/1]. 
NKM and MM acknowledge support from the Lamarr Institute for Machine Learning and Artificial Intelligence.
%


\bibliographystyle{icml2026}
\bibliography{ref}

\newpage
\appendix
\crefalias{section}{appendix}

\section*{Appendix}

This appendix entails further technical details and experiments to complement the main text. 
It includes:

\resizebox{1.\columnwidth}{!}{%
\centering
\begin{tabular}{p{5cm}|p{1.6cm}|p{1.3cm}}
\textbf{Title} & \textbf{Text} & \textbf{Appendix} \\
\hline
\hline
Derivation of the interference behind the Berstein-Vazirani algorithm & \Cref{sec:warm_up} & \Cref{app:bv} \\
\hline
Proof of the Fourier Expansion Theorem &\Cref{theo:fourier_bool}&\Cref{app:proof_fourier}\\ 
\hline
Derivation of BVNs from a Linear Algebra Perspective &\Cref{theo:fourier_bool}&\Cref{app:linear_algebra}\\ 
\hline
Layered Expressive Representation&\Cref{sec:expressive_repr}&\Cref{app:layered_repr}\\
\hline
Rectangle Expressive Representation&\Cref{sec:expressive_repr}&\Cref{app:rectangle_repr}\\
\hline
Interference Operators in Comparison&\Cref{sec:interference_ops}&\Cref{app:interference_ops} \\
\hline
Complexities and Systematic Comparison &\Cref{sec:exterimental_results}&\Cref{app:complexities_and_comparision}\\
\hline
Inspecting the sampled Coefficients&\Cref{sec:2d_classification}&\Cref{app:inspection}\\
\hline
Comparing Filled and Unfilled Training Data&\Cref{sec:2d_classification}&\Cref{app:filled_comparison}\\
\hline
Future Work and Technical Extensions& \Cref{sec:conclusion}&\Cref{app:future_work}\\
\hline
\end{tabular}
}

\section{Interference in Bernstein--Vazirani}
\label{app:bv}
We shortly review the Bernstein--Vazirani algorithm \cite{bernstein1993quantum} and explain how interference occurs.

From the three steps of the Bernstein--Vazirani algorithm (\Cref{eq:bv1,eq:bv2,eq:bv3}) and after replacing $f$ by its signature $f(x) = s \odot x$, we have
\begin{align}
    \ket{\psi_1} &= \frac{1}{\sqrt{2^n}} \sum_{x \in \Z_n^2} \ket{x} \\
    \ket{\psi_2} &= \frac{1}{\sqrt{2^n}} \sum_{x \in \Z_n^2} (-1)^{f(x)} \ket{x} \\
    \ket{\psi_3} &= \sum_{y \in \Z_n^2} \widehat{f}(y)\ket{y}.
\end{align}
with Fourier coefficients
\begin{align}
    \widehat{f}(y) 
    &= \frac{1}{2^n} \sum_{x \in \Z_n^2} (-1)^{s\odot x+y\odot x}\\
    &=\frac{1}{2^n} \sum_{x \in \Z_n^2} (-1)^{(s\odot x+y\odot x) \bmod 2} \label{eq:phase}.
\end{align}
We want to show that $\widehat{f}(y) = 1$ if and only if $y=s$ and $0$ otherwise.
On $\Z_n^2$, we use two important properties:
\begin{align}
    (a + b)\hspace{-.3cm} \mod 2 &= (a\hspace{-.3cm} \mod 2 + b\hspace{-.3cm} \mod 2) \bmod 2, \label{eq:mod1}\\
    (a \cdot b)\hspace{-.3cm} \mod 2 &= (a\hspace{-.3cm} \mod 2) \cdot (b\hspace{-.3cm} \mod 2) \bmod 2\label{eq:mod2}.
\end{align}

After applying \Cref{eq:mod1,eq:mod2} to the phase term in \Cref{eq:phase}, it holds
\begin{align}
    &(s\odot x + y\odot x) \bmod 2 \\
    &= (s^\top x \bmod 2 + y^\top x \bmod 2) \bmod 2 \\
    &= ((s \bmod 2)^\top (x \bmod 2) \nonumber \\
    &\qquad + (y \bmod 2)^\top (x \bmod 2)) \bmod 2 \\
    &= ((s \bmod 2 + y \bmod 2)^\top (x \bmod 2)) \bmod 2 \\
    &= (((s + y) \bmod 2)^\top (x \bmod 2)) \bmod 2 \\
    &= ((s \oplus y)^\top x) \bmod 2 \\
    &= (s \oplus y)^\top \odot x \label{eq:phase_only}.
\end{align}

Now substituting \Cref{eq:phase_only} in  \Cref{eq:phase}, we have
\begin{align}
    \frac{1}{2^n} \sum_{x \in \Z_n^2} (-1)&^{(s\odot x+y\odot x) \bmod 2} \nonumber \\
    &= \frac{1}{2^n} \sum_{x \in \Z_n^2} (-1)^{(s \oplus y)^\top \odot  x}
\label{eq:orthogonality1}\\
    &= \frac{1}{2^n} \sum_{x \in \Z_n^2} \prod_{i=1}^n (-1)^{(s_i \oplus y_i) \cdot  x_i}
\label{eq:orthogonality2}\\
    &= \frac{1}{2^n} \prod_{i=1}^n \sum_{x_i \in \Z_1^2} (-1)^{(s_i \oplus y_i) \cdot  x_i}
\label{eq:orthogonality3}\\
    &= \frac{1}{2^n} \prod_{i=1}^n (1 +(-1)^{(s_i \oplus y_i)})
\label{eq:orthogonality4}.
\end{align}
In the last term, we see that the complete sum evaluates to $0$ if $s_i \oplus y_i=1$ for any index $i$, and to $1$ if $s_i \oplus y_i = 0$ for all $i$.
On the other hand, we have $s_i \oplus y_i = 0$ if and only if $s_i = y_i$.

This collapsing sum to either $0$ or $1$ is called \emph{interference}, where undesired $y$ interfere destructively, while the desired solution interferes constructively.

\section{Proof of the Fourier Expansion Theorem}
\label{app:proof_fourier}

In this section, we provide the proof of the Fourier Expansion \Cref{theo:fourier_bool}.

\begin{proof}
The proof follows from linear algebra by expressing $f$ uniquely in an orthonormal basis via inner products.

First, we verify that $\{\chi_y\}_{y \in \Z_2^n}$ forms an orthonormal basis of the vector space $\mathcal{F} := \{f:\{0, 1\}^n \to \R\}$ of functions from $\Z_2^n=\{0, 1\}^n$ to $\R$. 
This is principally possible because any function in $\mathcal{F}$ can be replaced by a vector in $\R^{2^n}$ by enumerating all function values.
We define the inner product on $\Z_2^n$ as $\braket{f, g} = \mathbb{E}_{x \in \Z_2^n}[f(x)\,g(x)]$.

As derived in \Cref{eq:orthogonality1,eq:orthogonality2,eq:orthogonality3,eq:orthogonality4}, $\forall y_1, y_2 \in \Z_2^n$, it holds 
\begin{align}
    \langle \chi_{y_1}, \chi_{y_2} \rangle 
    &= \frac{1}{2^n} \sum_{x \in \Z_2^n} (-1)^{(y_1\odot x+y_2\odot x) \mod 2} \\
    &= 
    \begin{cases}
        1, & \text{if } y_1 = y_2, \\
        0, & \text{otherwise}.
    \end{cases}
\end{align}
Since there are $2^n$ functions $\chi_y$, matching the dimension of $\mathcal{F}$,  
the set $\{\chi_y\}_{y \in \Z_2^n}$ forms an orthonormal basis.

Hence, any $ f \in \mathcal{F}$ expands uniquely as  
\begin{equation}
    f(x) = \sum_{y \in \Z_2^n} \widehat{f}(y) \chi_y(x),
\end{equation}
where
\begin{equation}
    \widehat{f}(y) = \langle f, \chi_y \rangle = \frac{1}{2^n} \sum_{z \in \Z_2^n} f(z) \chi_y(z).
\end{equation}
\vspace{-.2cm}
\end{proof}

\section{Alternative Derivation of (Generalised) Bernstein--Vazirani Networks from a Linear Algebraic Perspective}
\label{app:linear_algebra}

Related, but not strictly equivalent concepts to our generalised Bernstein--Vazirani algorithm are \emph{generalised Fourier series} or \emph{Sparse dictionary learning}.

To (informally) generalise \Cref{theo:fourier_bool}, we assume that there is a set of functions $\mcal G = \{g_y\}_{y\in\mathbb{Z}_2^m} \subset \mcal F$, such that 
\begin{equation}
\label{eq:fourier_bool_gen}
    f(x) \approx \sum_{y \in \Z_2^m} \braket{f, g_y} g_y,
\end{equation} 
with $\braket{\cdot, \cdot}$, defined as $\braket{f, g_y} = \frac{1}{2^n} \sum_{x \in \Z_2^n}f(x)g_y(x)$, being the standard inner product in $\Z_2^n$.
We further still assume access to the function oracle $U_f$, which realises the state
\begin{equation}
\label{eq:psi}
    \ket{\psi} = \frac{1}{\sqrt{2^n}} \sum_{x \in \Z_2^n} f(x) \ket{x}.
\end{equation}

As our inputs $x$ live in $\R^{2^n}$ and we seek expressive models, we are interested in settings where $m \ge n$, i.e., where $\mcal G$ is an \emph{overcomplete basis} for the function space $\mcal F$.

Now, independently, let $\mcal Y := \{\ket{y}, y \in \Z_2^m\}$ be some orthonormal basis, with $\ket{y}$ written in the computational basis as
\begin{equation}
\label{eq:y}
    \ket{y} = \frac{1}{\sqrt{2^m}} \sum_{x \in \Z_2^n}\sum_{z \in \Z_2^{m-n}} a_y^{x,z} \ket{x,z},
\end{equation}
for some amplitudes $a_y^{x,z} \in \C$, where the $\ket{z}$ register is an ancillary register with $(m-n)$ qubits.  
We embed $\ket{\psi}$ from \Cref{eq:psi} in the $2^m$-dimensional Hilbert space as $\ket{\tilde{\psi}} = \ket{\psi,0}$.  
Expressing $\ket{\tilde \psi}$ in the basis $\mcal Y$ yields
\begin{align}
    \ket{\tilde \psi} &= \sum_{y \in \Z_2^m} \braket{y|\tilde \psi} \ket{y} \\
    &= \frac{1}{2^n} \sum_{y \in \Z_2^m} \sum_{p,q \in \Z_2^n} a_y^{p,0} f(q) \braket{p|q} \ket{y} \\
    \label{eq:psi_1}
    &= \frac{1}{2^n} \sum_{y \in \Z_2^m} \sum_{x \in \Z_2^n} a_y^{x,0} f(x) \ket{y}.
\end{align}

Now, to connect $\mcal Y$ and $\mcal G$, we observe that the amplitudes $a_y^{x,0}$ define the family of functions
\begin{equation}
    g_y(x) := a_y^{x,0}.
\end{equation}
Even though the initial state has ancilla qubits in $\ket{0}$, the \emph{different} $y$ basis vectors spread across the ancillary dimensions, giving rise to an \emph{overcomplete dictionary} of functions $g_y$.  
Constructing the basis $\mcal Y$ such that $g_y(x) = a_y^{x,0}$ transforms \Cref{eq:psi_1} into
\begin{equation}
    \ket{\tilde \psi} = \sum_{y \in \Z_2^m} \braket{g_y, f} \ket{y}.
\end{equation}
Measuring $\ket{\tilde \psi}$ in the basis $\mcal Y$ reveals, with probability $|\braket{g_y,f}|^2$, the corresponding $y$ and the associated $g_y$ needed to approximate $f$ according to \Cref{eq:fourier_bool_gen}.

If $n=m$ and the set $\mcal G$ is orthonormal in $\mcal F$, we recover the standard BV setting (up to the chosen basis $\mcal Y$).  
The generalised BV setting emerges when $m \ge n$, giving rise to an overcomplete and expressive basis.

\paragraph{How to design $\mcal Y$ that realises $\mcal G$?}
This is now an architectural design choice.  
The important point is that enlarging the system with ancillary qubits allows us to realise an overcomplete basis.  
The interference operator acting on the representation register effectively disentangles it from the input $x$, allowing us to write \Cref{eq:psi_1} as a sum over $x$ while isolating the $\ket{y}$ register.

\paragraph{How to measure in general bases?}
We are interested in measuring $\ket{\tilde \psi}$ in a basis $\mcal Y$ that is not necessarily the computational basis.  
Let $Y$ be the measurement observable of the basis $\mcal Y$.  
As a Hermitian operator, $Y$ can be diagonalised. 
There exists a unitary $U_Y$ such that
\begin{equation}
    U_Y Y U_Y^\dagger = D,
\end{equation}
where $D$ is diagonal in the computational basis.  
Then,
\begin{equation}
    \braket{\tilde \psi|Y|\tilde \psi} = \braket{\tilde \psi|U_Y^\dagger D U_Y|\tilde \psi}.
\end{equation}
Measuring $\ket{\tilde \psi}$ in the basis $\mcal Y$ thus reduces to implementing $U_Y$ followed by measurement in the computational basis.  
This is exactly the action of our expressive representation circuit followed by the interference operator.

\section{Layered Expressive Representation}
\label{app:layered_repr}

In the following, we provide a quantum implementation of a classical two-layer MLP with step activation functions and fixed biases as an exemplary expressive representation to be used in our generalised BVNs (see \Cref{sec:expressive_repr}).

\cref{eq:representation} is written compactly for the input, activation, and weights.
In practice, the system may be organised into several sub-registers to support different input dimensions and to implement a classical layered representation reversibly. 

\begin{definition}[Quantum Layer]
\label{def:quantum_layer}
A quantum layer $\ell$ consists of multiple units. 
Each unit is specified by an output register $\ket{0}$, an activation function $\sigma$, and $d$ weight registers $\ket{t}$, where $d$ is the dimension of the activation vector of layer $\ell-1$. 
Given an input activation register $\ket{\sigma_{\ell-1}}$, the layer implements the coherent update
\begin{equation}
    \ket{\sigma_{\ell-1}}\ket{0}\ket{t}
    \mapsto
    \ket{\sigma_{\ell-1}}\ket{\sigma_\ell}\ket{t},
    \ 
    \sigma_\ell = \sigma\big(\sum_{i=1}^d t_i\,(\sigma_{\ell-1})_i\big).
\end{equation}
Ancillary qubits are abstracted but may be required for reversible arithmetic.
\end{definition}

\begin{figure}[!ht]
  \centering
    \includegraphics[]{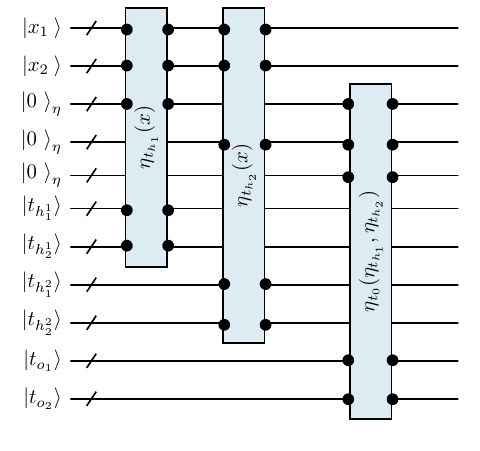}
  \caption{
  A quantum circuit realising an expressive representation of the input.
It implements a fully connected two-layer network.
The input is a two-dimensional register $\ket{x_1}\ket{x_2}$.
The hidden layer consists of two units, each computing $\eta_{t_{h_j}}(x) =\sigma (t_{h_j^1}x_1 +  t_{h_j^2}x_2)$
and the output unit computes $\eta_o$ in the same way, but acting on the hidden activations instead of the input.
The circuit is implemented using quantum arithmetic with ancillary qubits (not shown).
  }
\label{fig:two_layer_repr}
\end{figure}

In \cref{fig:two_layer_repr}, we illustrate a quantum circuit implementing an expressive representation of the input.
It realises a fully connected two-layer network.
The 2-dimensional input $\ket{x}$ is stored in two sub-registers, forming the input layer.
The hidden layer contains two units, each with two weight registers that modulate the input; the weighted sum is passed through an activation and written to an activation register.
The output layer is a single unit acting on the hidden activations.
For this example, the decomposition
$\ket{x}\ket{\eta_t(x)}\ket{t}$ 
from \cref{eq:representation} becomes
\begin{align}
    \ket{x} 
    &= \ket{x_1}\ket{x_2},\\
    \ket{\eta_t(x)} 
    &= \ket{\eta_{t_{h_1}}(x)}\ket{\eta_{t_{h_2}}(x)}\ket{\eta_{t_{0}}(\eta_{t_h})},\\
    \ket{t}
    &= \ket{t_{h_1^1}}\ket{t_{h_1^2}}\ket{t_{h_2^1}}\ket{t_{h_2^2}}\ket{t_{o^1}}\ket{t_{o^2}},
\end{align}
where 
$
    \eta_{t_{h_j}}(x) =\sigma \left(t_{h_j^1}x_1 + t_{h_j^2}x_2\right),
$
and similarly for $\eta_o$ acting on the hidden activations.

\paragraph{2-Layer Representation with Step Activation Function.}
We experimented with the architecture shown in \cref{fig:two_layer_repr} in the 2D case, but due to the induced circuit size, we were restricted to a two-layer model with two hidden units and one output unit, using binary weights and no biases.
The circuit computes the weighted sum and activates the output qubit (sets it to 1) if this sum does not exceed a hard-coded threshold.
This requires three qubits per unit: one weight qubit for each of the two input dimensions and one output qubit.
Such a binary 2-layer MLP therefore has six weights in total, corresponding to $2^6 = 64$ possible network configurations.

In \cref{fig:augmented_basis}, we visualise the possible activations produced by this binary model across all weight configurations.
The dimensional qubit-resolution of $x$ is $6$, and the thresholds of the hidden and output units are set to $32$ and $1$, respectively, in one setting, and to $64$ and $1$ in the other.
Although the network has 64 possible weight assignments, the resulting activations are not very expressive: only seven distinct partitions of the 2D input space appear in the first setting and three in the second.
The model would require learnable biases to shift activation boundaries freely, but introducing biases would significantly increase the qubit count.
In addition, the hidden activations are no longer needed once they have been used in the output layer.
Hence, they unnecessarily grow the qubit count.

\begin{figure}[ht]
    \centering
    \includegraphics[width=.8\linewidth]{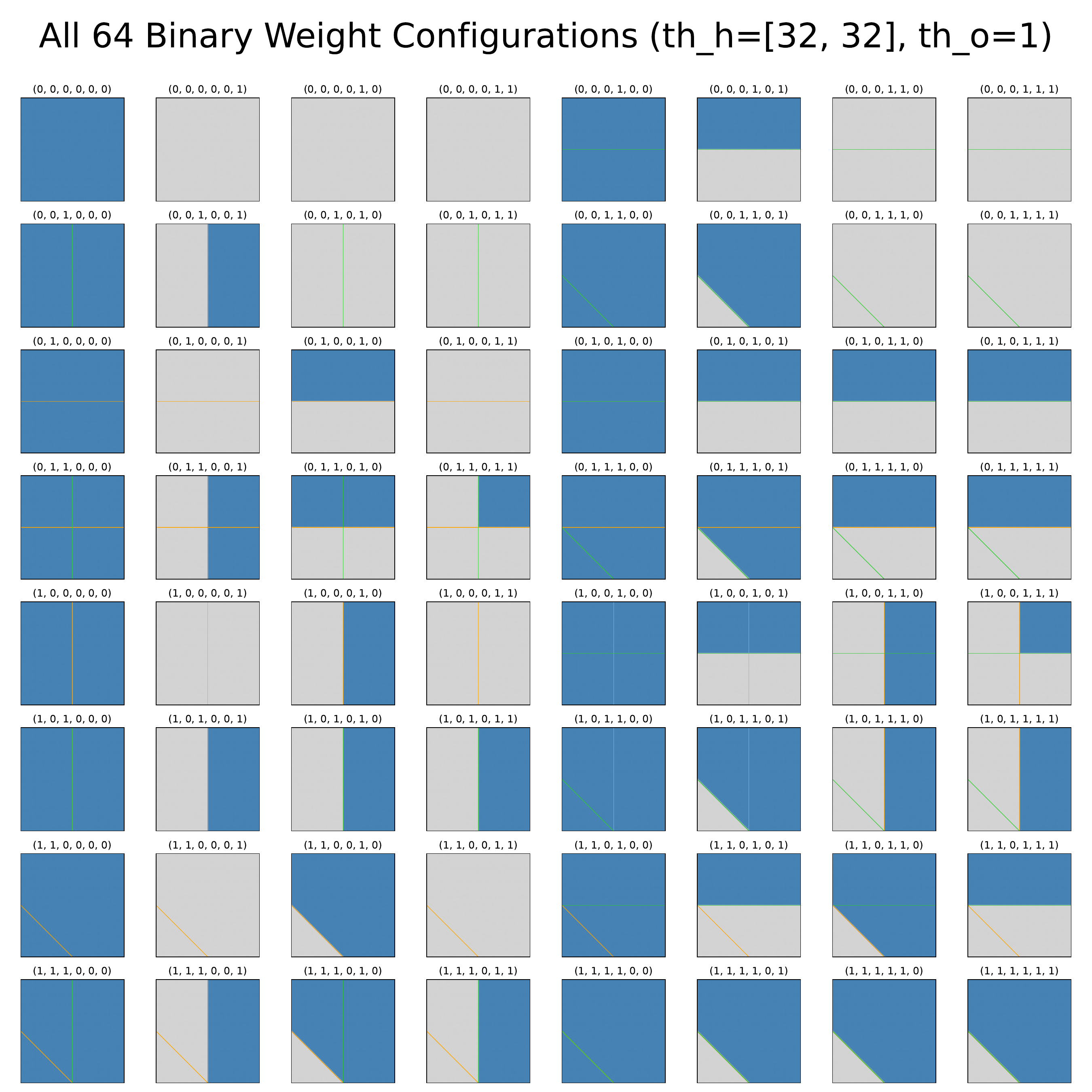}
    \hspace{1cm}
    \includegraphics[width=.8\linewidth]{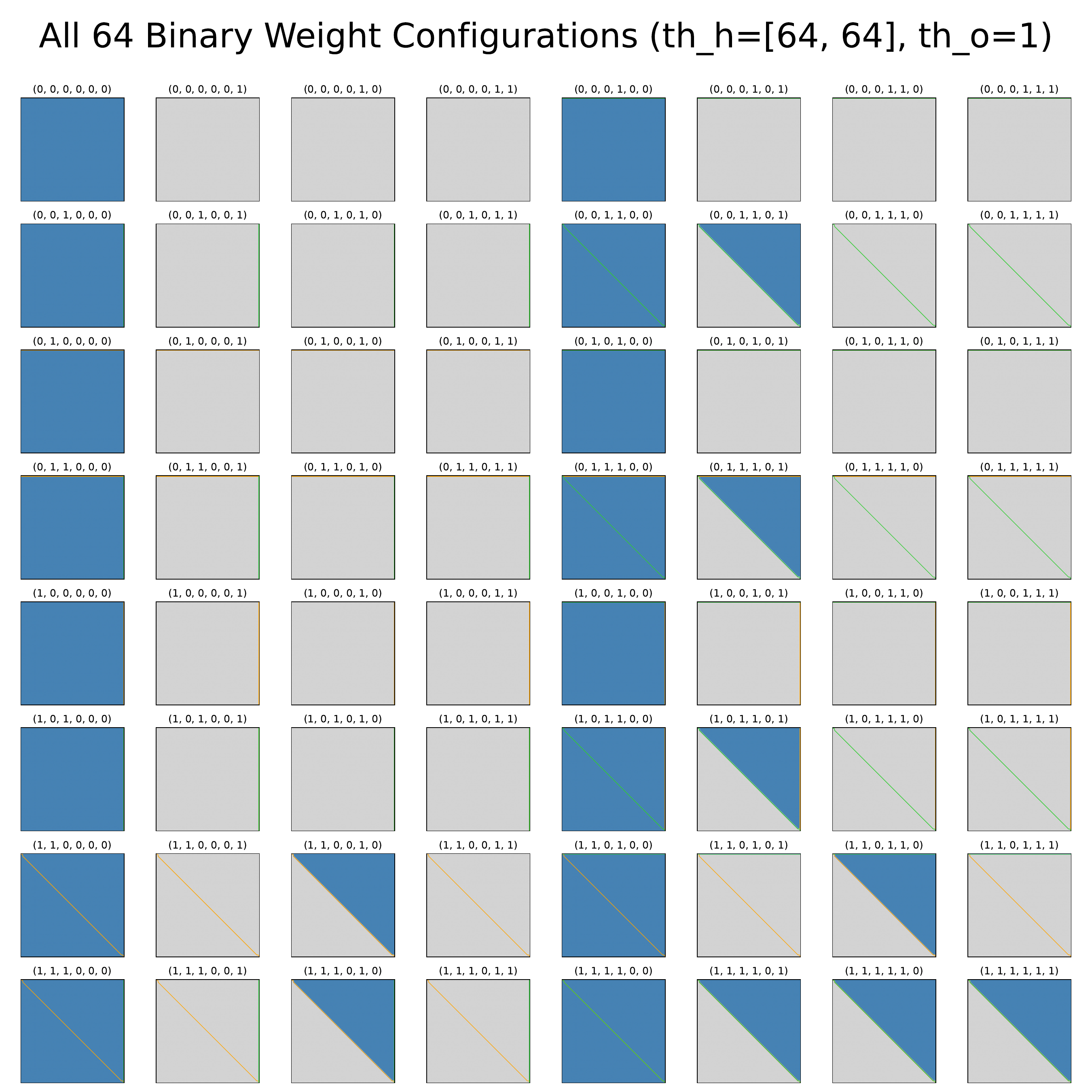}
    \caption{
    All distinct activation patterns produced by the 2-layer binary MLP for all 64 possible weight configurations.  
    Top: thresholds (32, 1) for hidden and output units.  
    Bottom: thresholds (64, 1).  
    Despite the large number of weight settings, only a small number of unique input-space partitions are realisable, highlighting the limited expressivity of this bias-free architecture.
    Blue means $1$ and grey $0$.
    }
    \label{fig:augmented_basis}
\end{figure}

In \cref{fig:shots_img_repr_sigmoid}, we visualise the image fitting results using the 2-layer MLP as the expressive representation, with thresholds $(16, 1)$ in the hidden and output layers. 
We observe that the MLP with fixed thresholds still lacks expressivity and does not significantly improve over the standard BVN.

\begin{figure}[ht]
    \centering
    \includegraphics[width=1\linewidth, trim={0cm 0cm 0cm .9cm}, clip]{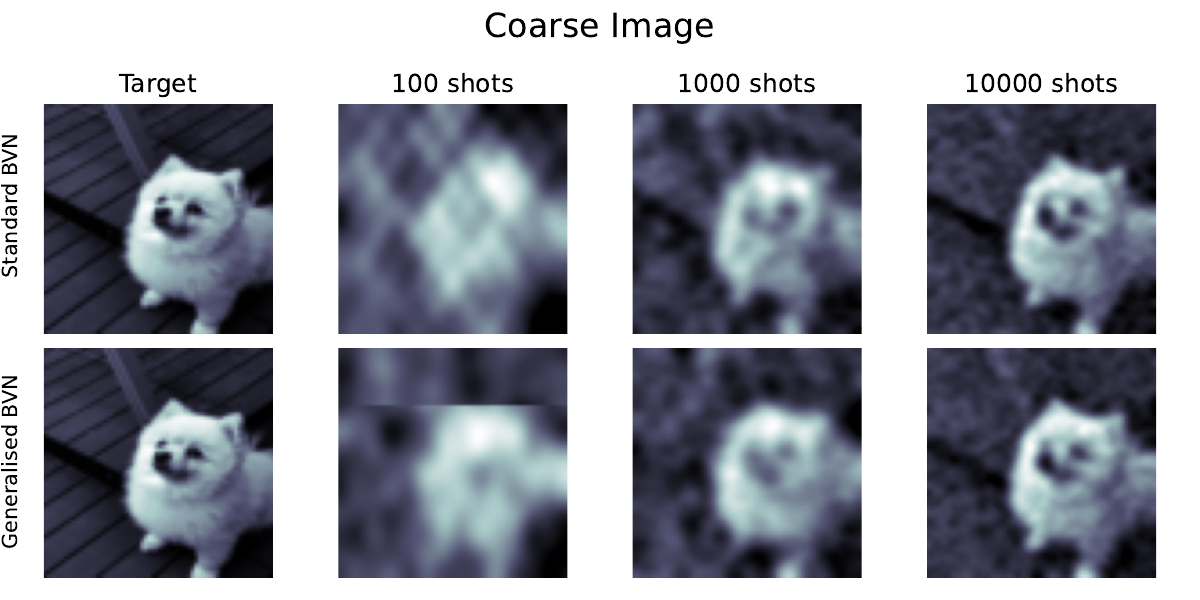}
    \caption{
    Visualising the fitting results of the standard (top row) and generalised (bottom row) BVNs for different numbers of shots, using the 2-layer MLP as the expressive representation.  
    With few shots, we observe the effect of the MLP, which activates different regions of the input domain.  
    However, the fixed region boundaries limit the model's ability to capture fine details in the image.
    }
    \label{fig:shots_img_repr_sigmoid}
\end{figure}

\section{Rectangle Expressive Representation}
\label{app:rectangle_repr}

In the following, we provide details on the rectangle activation as an exemplary expressive representation to be used in our generalised BVNs (see \Cref{sec:expressive_repr}).

To overcome the limited expressivity of bias-free 2-layer MLPs, we introduce the \emph{rectangle activation}.  
Each input dimension is represented with a binary register, and smaller mask registers define local subdomains.  
The circuit accumulates contributions from each dimension into a shared sigma register, with superposition allowing all inputs and rectangles to be processed simultaneously.

\begin{figure}[h]
  \centering
    \includegraphics[width=.5\textwidth]{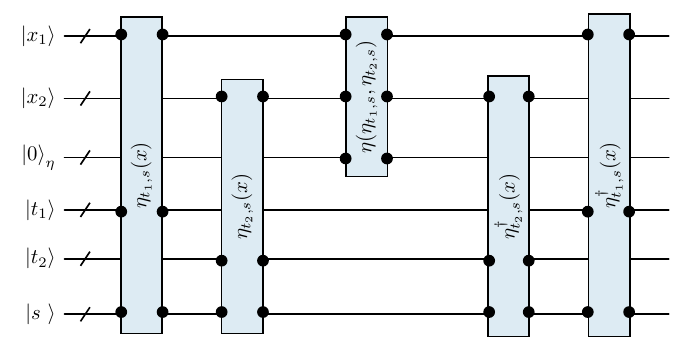}
  \caption{
  Quantum circuit implementing the rectangle activation.  
    We circularly shift each input dimension $x$ by its corresponding rectangle parameter $t$, and activate the output qubit if all high-significance qubits of $x$ are 0, which corresponds to activating the rectangle parametrised by $t$.  
    To reduce artefacts, the superposition register allows for overlapping rectangles by further shifting the input conditionally.
  }
\label{fig:mask_repr}
\end{figure}

In \Cref{fig:mask_repr}, we illustrate the circuit implementing the rectangle activation.  
Since hidden activations are only temporarily needed for computing the output, we use the input register to store them temporarily.  
Let $r_x$ be the qubit resolution of one dimension $x_i$ of the input $x$ and $r_t$ that of the rectangle positions.  
Then, the rectangle width is $2^{r_x - r_t}$, and there are $2^{r_t}$ positions to place the rectangle along the considered dimension $i$.
We circularly shift each input dimension $x_i$ by its corresponding parameter $t_i$:
\begin{equation}
    \ket{x_i} \gets \ket{x_i + t_i \cdot 2^{r_x - r_t} \mod 2^{r_x}}.
\end{equation}
The output qubit is activated (set to 1) if all $(r_x - r_t)$ most significant qubits of the input register are $0$.
After the activation, the shift is undone, leaving the input register back in the $x$-state.  
To reduce artefacts, we allow rectangles to overlap by introducing a superposition register $\ket{s}$, which duplicates rectangle parameters.  
Depending on the value of $s$, we shift the input further before activation.
In our experiments, we use one superposition qubit and shift by half of the rectangle resolution if the superposition bit is 1, yielding
\begin{equation}
    \ket{x_i} \gets \ket{x_i + t_i \cdot 2^{r_x - r_t} + s \cdot 2^{r_t/2} \mod 2^{r_x}}.
\end{equation}

In \Cref{fig:augmented_basis_mask}, we visualise the rectangle activations in 2D. 
The dimensional qubit-resolution of $x$ is $5$, and that of the rectangle positions is $3$.
With only $6$ additional qubits (compared to nine in the MLP case plus ancillaries), the procedure effectively creates localised rectangles in the input space where the activation is turned on, modifying the standard BV basis locally.  
The result is a highly expressive activation pattern that captures fine-grained variations in the input while remaining resource-efficient.

\begin{figure}[ht]
    \centering
    \includegraphics[width=.8\linewidth]{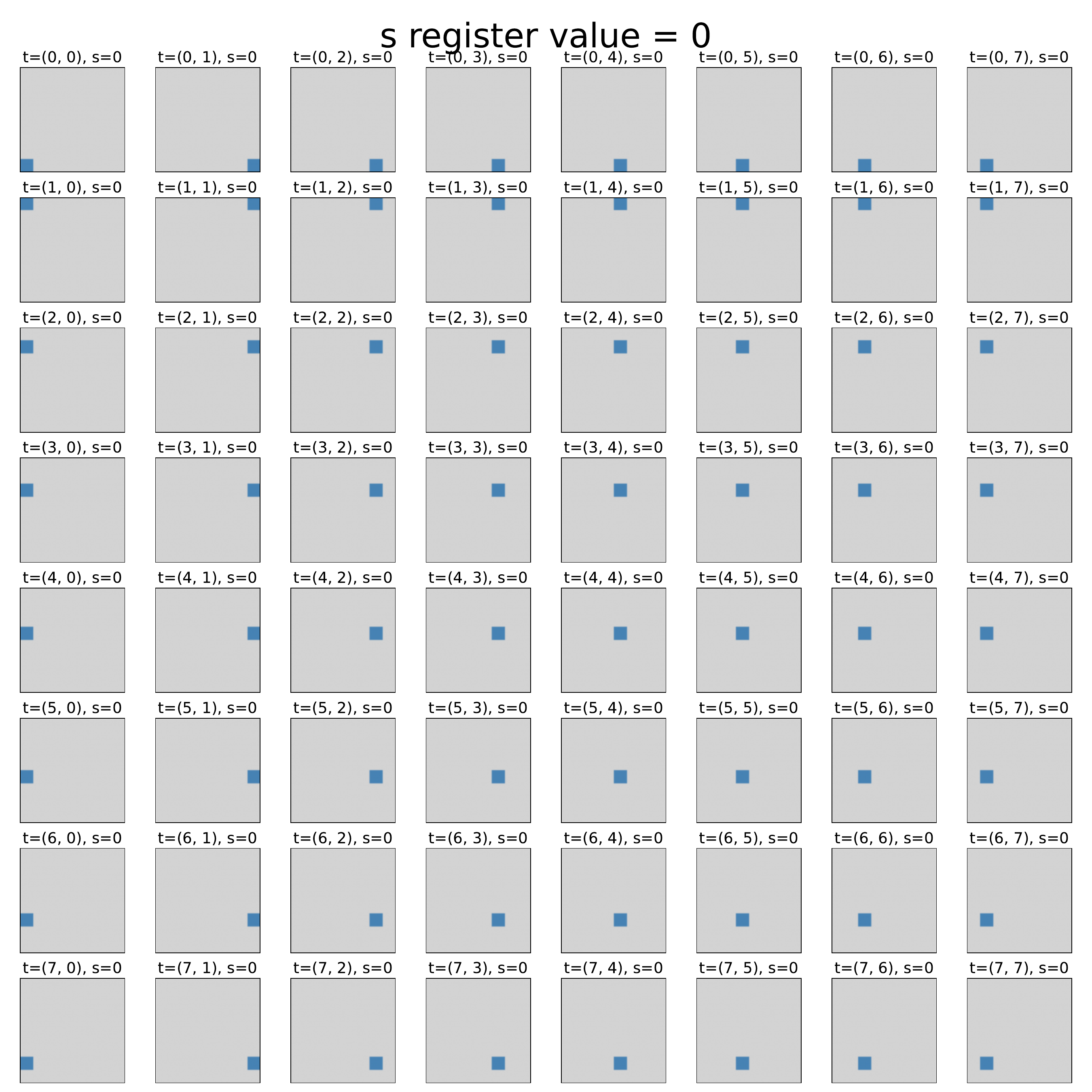}
    \hspace{1cm}
    \includegraphics[width=.8\linewidth]{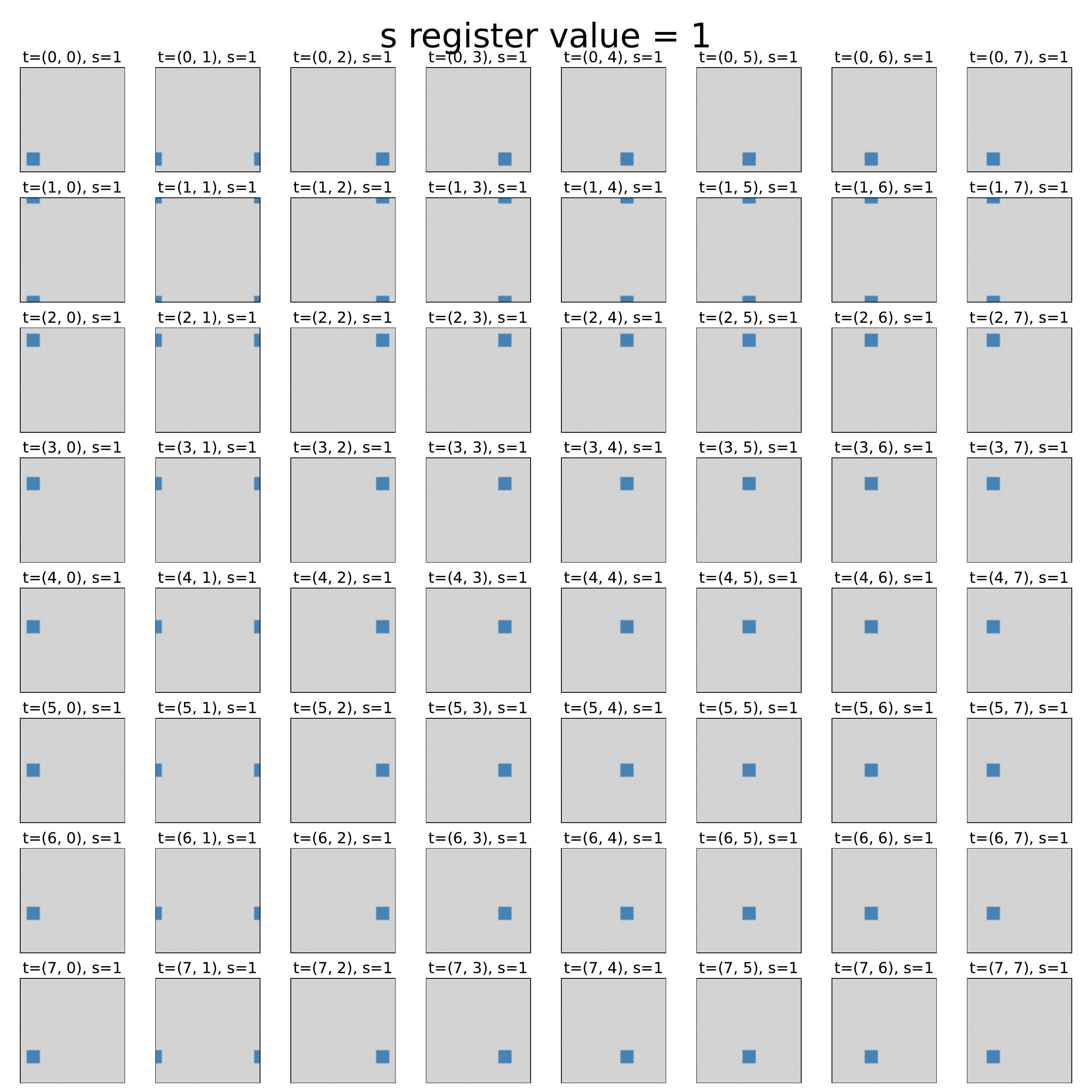}
    \caption{
    All distinct activation patterns produced by the rectangle activation for all $64$ possible weight configurations.  
    Top: activations for $s = 0$.  
    Bottom: activations for $s = 1$.  
    Each weight configuration generates a distinct local activation, illustrating the expressive power of the rectangle activation without requiring biases or additional qubits. 
    Blue indicates $1$ and grey indicates $0$.
    }
    \label{fig:augmented_basis_mask}
\end{figure}

In \cref{fig:shots_img_repr_mask}, we validate the expressivity of the rectangle activation for the resolutions $r_r = 3$ and $r_r = 5$, and ablate the number of shots used in the image fitting experiment.  
We clearly see how each rectangle locally modifies and corrects the standard BV basis functions toward a better fit of the target.  
As a consequence, the generalised model converges faster and captures finer details than the standard model.
The finer the resolution $r_r$, the better the results.

\begin{figure}[t]
    \centering
    \includegraphics[width=1\linewidth, trim={0cm 0cm 0cm .9cm}, clip]{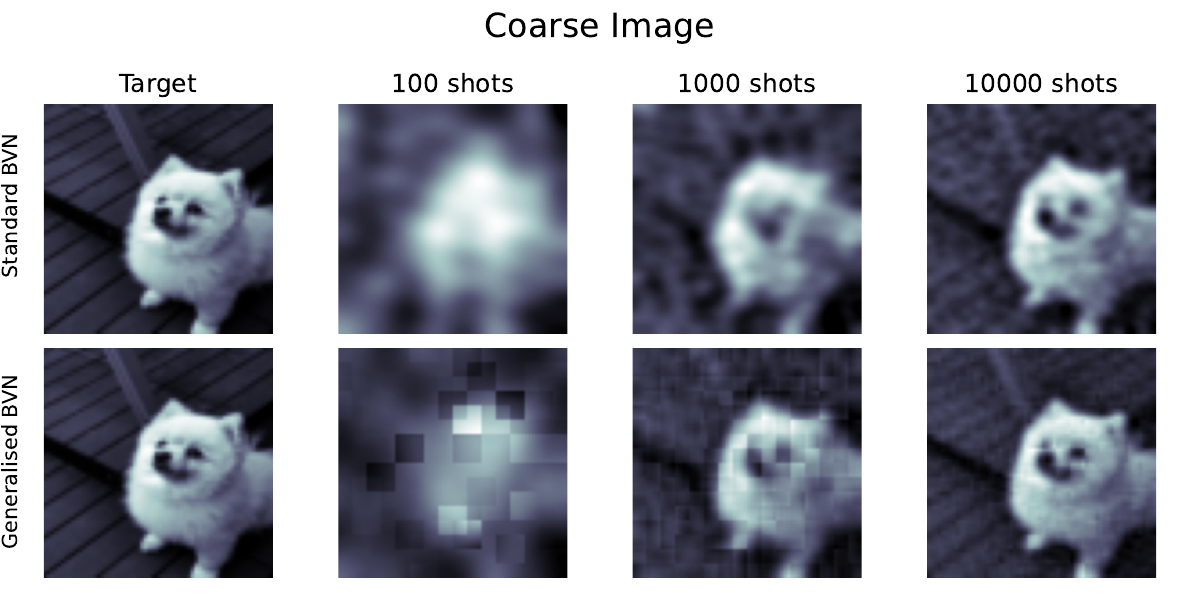}
    \includegraphics[width=1\linewidth, trim={0cm 0cm 0cm .9cm}, clip]{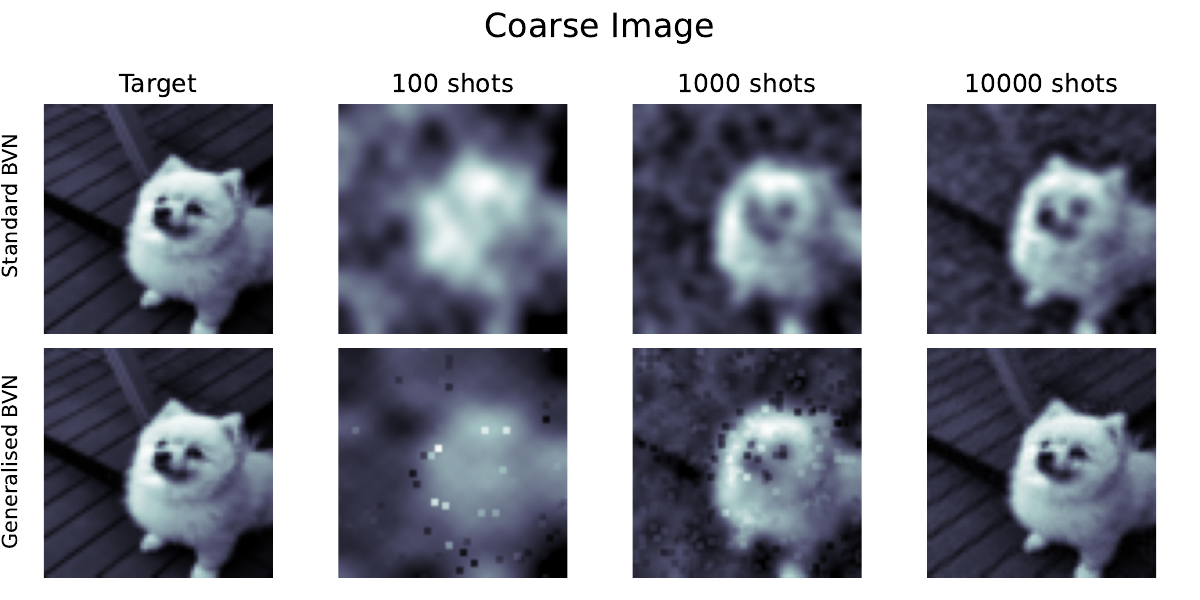}
    \caption{
    Visualising the fitting results of the standard and generalised BVNs for different numbers of shots and two rectangle resolutions, 3x3 (top) and 5x5 (bottom).
    With a few shots, we observe the effect of the rectangle activation functions.  
    Each rectangle locally modifies the standard BV basis, allowing the model to better fit the target in that region.  
    With a sufficient number of shots, the rectangles are no longer visible, and the generalised model becomes more accurate than the standard one.
    The finer rectangle resolution outperforms the coarse.
    }
    \label{fig:shots_img_repr_mask}
\end{figure}

\section{Interference Operators in Comparison}
\label{app:interference_ops}
We visualise the activation or basis function induced by the Hadamard, Fourier and Chebyshev interference operators in \Cref{fig:interference_ops}.
The Hadamard operator produces step parity functions with amplitude $1$, while the Fourier and Chebyshev operators introduce finer variations.

\begin{figure}[ht]
    \centering
    \includegraphics[width=1.\linewidth]{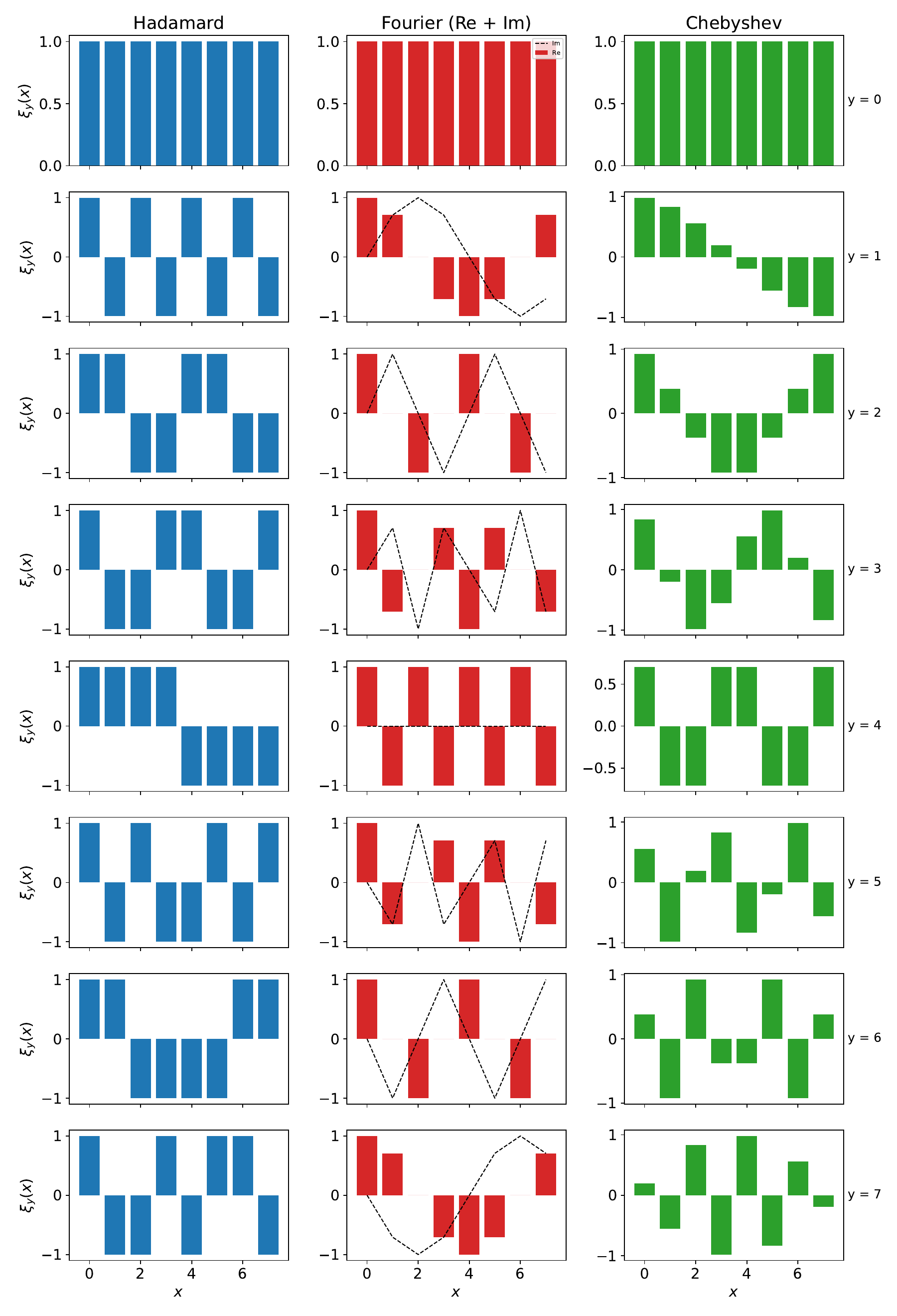}
    \caption{Interference operators in comparison.
    We visualise the basis functions $\xi_y(\cdot)$ for the Hadamard, the Fourier and the Chebyshev operators.
    The Hadamard basis functions are parity activations that yield $\pm1$ depending on whether the weighted sum of the input is odd or even.
    The Fourier and Chebyshev basis functions, in contrast, allow more variations than $\pm1$ and smoother activations. 
    }
    \label{fig:interference_ops}
\end{figure}

We test these interference operators on regression examples of fitting 1D functions with increasing frequencies.
The functions are generated by sampling random vectors of different lengths corresponding to the frequencies and interpolating the vectors.
We also include a step function to analyse the Gibbs phenomenon.
We compare the Hadamard and Chebyshev interference operators, along with the regularisation factor $\lambda$ of the coefficient reconstruction.
We omit the Fourier transform due to the imaginary part it introduces, which is not needed for the real-valued experiment.

\begin{figure*}[!ht]
    \centering
    \includegraphics[width=0.8\linewidth]{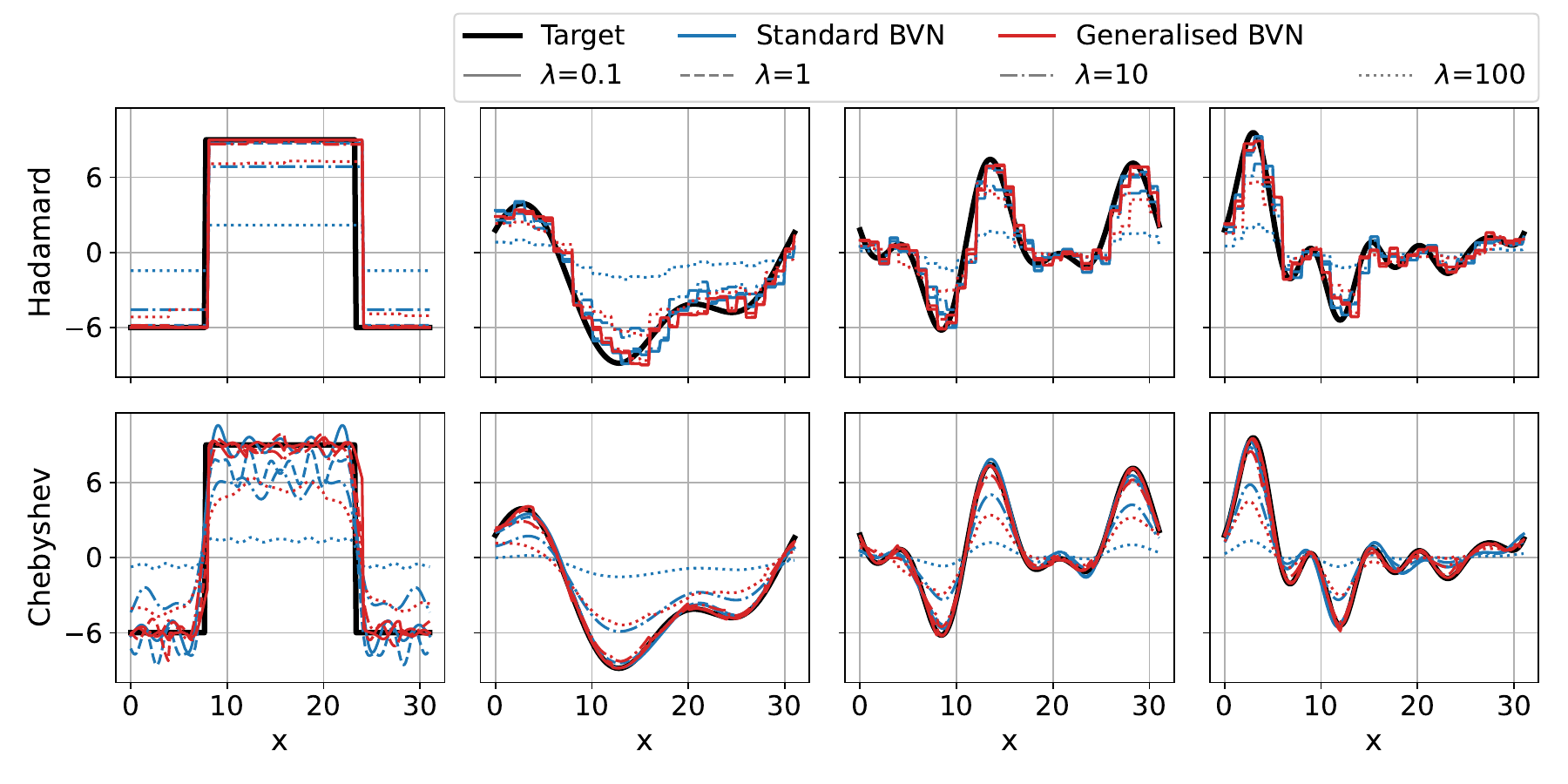}
    \caption{
    Regression of random 1D functions with different frequencies using the proposed interference-based quantum networks.  
    The top and bottom rows show results with the Hadamard and Chebyshev interference operators, respectively.  
    We compare the standard Bernstein--Vazirani model (blue) with the generalised model (red). 
    Within each model, we contrast the Gram inversion regularisation factors $\lambda$.
    The Chebyshev interference produces smooth functions as it acts in the real domain, while Hadamard, due to rounding, produces step functions as it acts in the binary domain.
    The generalised method with $\lambda=0.1$ fits the target function better than that of the standard model (see step function in column 1). 
    Increasing $\lambda$ regularises the approximation.
    }
    \label{fig:1d_regression}
\end{figure*}

\Cref{fig:1d_regression} presents the 1D regression results.
The Hadamard and Chebyshev operators perform similarly overall, as both can capture the function’s variations.
However, the Chebyshev operator, which operates in the real domain, produces smooth outputs, whereas the Hadamard operator yields step‑like functions due to rounding effects inherent to the binary domain.
Both the standard and generalised models fit the target function reasonably well.
Different values of $\lambda$ smooth the approximation, a property we expect to be beneficial in the presence of noise.
A closer inspection shows that the generalised model is more accurate and better captures high‑frequency variations.
For example, with the Chebyshev interference operator applied to the step function (first column, second row), the standard BVN oscillates around the target, whereas the locally corrected basis of the generalised model with rectangles significantly reduces these oscillations and matches the target more closely.
See also the image fitting results in \Cref{fig:shots_img_repr_mask}.

\section{Computational Complexity} 
\label{app:complexities_and_comparision}
Our BVNs expect training inputs (binary strings) to be placed in superposition and labelled via an oracle.
The gate complexity of this encoding depends on the sparsity of the training set relative to the dimension of the Hilbert space spanned by the input resolution; it is $O(mn)$ for $m$ being the effective size of the training set and $n$ the number of qubits (details in \Cref{app:complexities}).
Further, our rectangle representations with qubit resolution $n$ and $n_t$ for the input and parameter registers use $n_t$ controlled adders and their inverses, along with a Toffoli gate to compute and mark the rectangles in the $\eta$-register, yielding a gate complexity of $O(nn_t)$.
The Hadamard interference operator can be implemented with $O(n)$ gates on $n$ qubits, whereas Chebyshev \cite{williams2023quantum} requires $O(n + n^2)$ gates.

Lastly, for reconstructing the coefficients, we build the matrix $X \in \mathbb{R}^{m \times k}$ by evaluating each of the $k$ sampled basis functions on each of the $m$ training inputs. 
We then recover the coefficients in closed form as described in \Cref{eq:coefs}.
Let $c$ denote the cost of evaluating a basis function on a training input, which is $O(n)$ for the Hadamard interference operator and $O(1)$ for the Chebyshev and Fourier interference operators. 
Then the total cost of reconstructing the coefficients is:
\begin{itemize}[leftmargin=*,itemsep=0.1em,topsep=0.1em]
    \item $O(cmk)$ for constructing $X$,
    \item $O(mk^2)$ for computing the Gram matrix $G = X^\top X$,
    \item $O(mk)$ for computing $X^\top F$,
    \item $O(k^3)$ for solving the linear system $(G + \lambda I)^{-1} X^\top F$.
\end{itemize}
In practice, this reconstruction step was not prohibitively expensive. 
Its runtime was comparable to a single training epoch in classical models, as it involves a single pass over the training set.

\subsection{Dataset-Encoding and Oracle Complexities}
\label{app:complexities}
Since data encoding is a major bottleneck in QML, we verify that this step is efficient in our BVNs.
In our experiments, dataset and oracle steps are combined into an amplitude encoding of a state vector with amplitudes proportional to the label values.
As an example, for Iris classification, a $4$-dimensional scaled and rounded feature $x = (7,6,10,11)$ is encoded with $4$ qubits per dimension as $x = (0111, 0110, 1010, 1011)$, giving when concatenated $x = 0111011010101011$. 
Each input is thus a $16$-qubit computational basis vector $\ket{x} \in \mathbb{C}^{2^{16}}$, and the labelled training set is encoded as a normalised superposition $\sum_x f(x) \ket{x}$ using amplitude encoding.  
The gate complexity of amplitude encoding depends on the sparsity of the input vector \cite{schuld2018supervised}.
Sparse-state preparation techniques have a complexity of $O(mn)$ for $m$ nonzero entries \cite{gleinig2021efficient,ramacciotti2024simple}, i.e., the effective size of the training set.
In practice, $m \ll 2^{n}$ ($150$ Iris, $333$ Penguins samples for a $2^{16}$ dimensional space), yielding highly sparse vectors.
Even for our dummy‑filled datasets in the Iris and Penguins experiments, the algorithm in \cite{mozafari2022efficient}, based on Decision Diagrams (DDs) is expected applicable, which complexity $O(kn)$, where $k$ is the number of DD paths.
Our dummy‑filled state is dominated by a constant‑amplitude component over large regions of the computational basis, with only a few data‑dependent deviations, so the effective path count $k$ is expected to scale mainly with $m$.

\emph{The data-loading cost in BVNs compares strictly better than variational QML methods, requiring $O(mn)$ separate state preparation operations per parameter-shift rule evaluation of a single partial derivative, multiplied by the number of shots required for each estimate, multiplied by the number of parameters, and multiplied by the number of training iterations.
}
Concretely, let $c$ denote the number of shots needed to estimate an expectation value, which is required for computing partial derivatives when using the only certified method, the parameter‑shift rule. 
To estimate an expectation value to precision $\epsilon$, it is known that the required number of shots scales as $c = 1/\epsilon^2$. 
For illustration, a single partial derivative with accuracy $\epsilon = 0.001$ already requires $2c = 2 \cdot 10^6$ circuit evaluations.
And we have to multiply this cost by the number of parameters and then by the number of training iterations.

\subsection{Sample Complexity Analysis}
\label{app:sample_complexities}
The number of measurement shots in BVNs is not directly transferable to the approximation accuracy of the target function $f$, and is also not influenced by the dimension of the input space.
Rather, it is determined by the ``sparsity'' of the target function $f$ in the chosen basis $\chi$. 

We can generalise the error bound analysis for the proposed BVN framework.
In fact, since the standard BVN is a special case of the generalised BVN, we can write $f$ uniquely in the standard subspace of the extended basis as
\begin{equation}
	f(x)=\frac{1}{\sqrt{2^n}}\sum_{y\in\Z_2^n}\widehat f(y,0,0)\,\chi_{y00}(x),
\end{equation}
or, more generally, as the non-unique expansion
\begin{equation}
	f(x)=\frac{1}{\sqrt{2^n}}\sum_{a}\widehat f(a)\chi_a(x),
\end{equation}
where $a=(y,z,t)$ indexes the generalised basis functions.

Now let $\Omega$ denote the set of sampled basis states of the generalised BVN, and define
\begin{equation}
	f_\Omega(x)=\frac{1}{\sqrt{2^n}}\sum_{a\in\Omega}\widehat f(a)\chi_a(x).
\end{equation}
Then the approximation error is
\begin{align}
	\epsilon
	&=
	\mathbb E_x\left[(f(x)-f_\Omega(x))^2\right] \\
	&=
	\frac{1}{2^n}
	\sum_x
	\left(
	\frac{1}{\sqrt{2^n}}
	\sum_{a\notin\Omega}
	\widehat f(a)\chi_a(x)
	\right)^2 \\
	&=
	\frac{1}{2^n}
	\sum_x
	\sum_{a,b\notin\Omega}
	\widehat f(a)\widehat f(b)
	\chi_a(x)\chi_b(x) \\
	&=
	\sum_{a,b\notin\Omega}
	\widehat f(a)\widehat f(b)
	\underbrace{
		\left(
		\frac{1}{2^n}
		\sum_x
		\chi_a(x)\chi_b(x)
		\right)
	}_{=:G_{ab}},
\end{align}
where $G$ is the Gram matrix of the sampled basis functions. Hence,
\begin{equation}
	\epsilon
	=
	\widehat F_{\bar\Omega}^{\top}
	G_{\bar\Omega}
	\widehat F_{\bar\Omega},
\end{equation}
where $\bar\Omega$ denotes the complement of $\Omega$. 
Since $G_{\bar\Omega}$ is positive semidefinite,
\begin{equation}
	\epsilon
	\le
	\lambda_{\max}(G_{\bar\Omega})
	\|\widehat f_{\bar\Omega}\|_2^2,
\end{equation}
where $\lambda_{\max}(G_{\bar\Omega})$ denotes the largest eigenvalue of the Gram matrix of the generalised basis functions and $\widehat F_{\bar\Omega} = (\widehat f(a))_{a \notin \Omega}$ is the vector of residual coefficients.

From the above derivation, we can further derive an explicit bound for the standard case. 
Indeed, for the standard BVN basis, the Gram matrix satisfies $G=I$ by orthonormality, and therefore $\epsilon = \|\widehat F_{\bar\Omega}\|_2^2 = 1-s$, where $s=\sum_{a\in\Omega}|\widehat f(a)|^2$ denotes the accumulated spectral mass. 
Thus, in the standard BVN, the approximation error is determined entirely by the unsampled Fourier spectrum.

For the generalised BVN, the induced basis is generally overcomplete and therefore non-orthogonal. 
Consequently, the approximation error satisfies $\epsilon \le \lambda_{\max}(G_{\bar\Omega}) \|\widehat F_{\bar\Omega}\|_2^2$.
This bound exposes a natural trade-off introduced by the expressive representation. 
The more correlated the basis functions are, the higher the penalty factor $\lambda_{\max}(G_{\bar\Omega})$. 
At the same time, the richer representation $\eta_t$ enables the target function to admit more compact representations in the extended basis, so that the norm of the residual coefficients $\|\widehat F_{\bar\Omega}\|_2^2$ may decrease substantially. 

Our empirical results of shots-vs-accuracy in \Cref{fig:random_ablation} show consistently lower approximation errors for the generalised BVNs, suggesting that this reduction in $\|\widehat F_{\bar\Omega}\|_2^2$ dominates the additional Gram-matrix factor.

\section{Inspecting the Sampled Coefficients}\label{app:inspection} 
We inspect the histograms of the coefficients sampled by the standard and generalised BVNs for the classification task on the 2D synthetic datasets \Cref{fig:histograms_2d_true,fig:histograms_2d_false}, as well as on the real-world Iris and Penguins datasets \Cref{fig:histograms_4d}. 

From the histograms of the standard BVN \Cref{fig:histograms_2d_true,fig:histograms_4d}, one can infer how easily the objective function can be approximated on a given dataset.  
Sparse histograms indicate that the target function is sparse in the BV basis, i.e., relatively easy to approximate with a finite number of shots.  
More challenging datasets, such as \texttt{spiral}, exhibit more dispersed histograms.

The generalised model \Cref{fig:histograms_2d_false,fig:histograms_4d} does not produce sparse histograms. 
This is expected, as it extends the standard BV basis and generates multiple copies of the dominant BV functions, each with a local modification.

\begin{figure}[!h]
\centering
    \begin{subfigure}{0.9\linewidth}
        \centering
        \includegraphics[width=\linewidth]{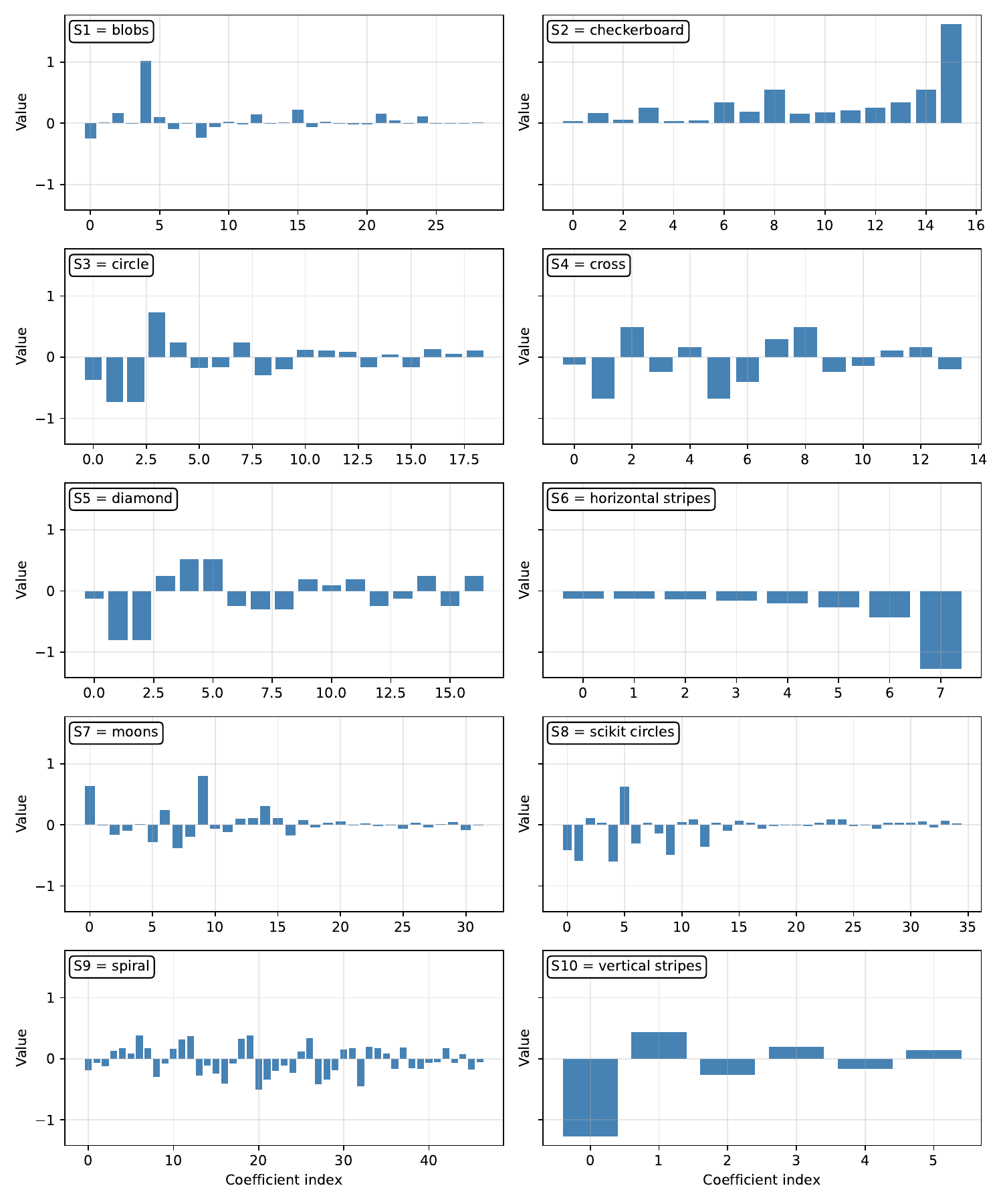}
        \caption{Standard BVN}
        \label{fig:histograms_2d_true}
    \end{subfigure}
    \begin{subfigure}{0.9\linewidth}
        \centering
        \includegraphics[width=\linewidth]{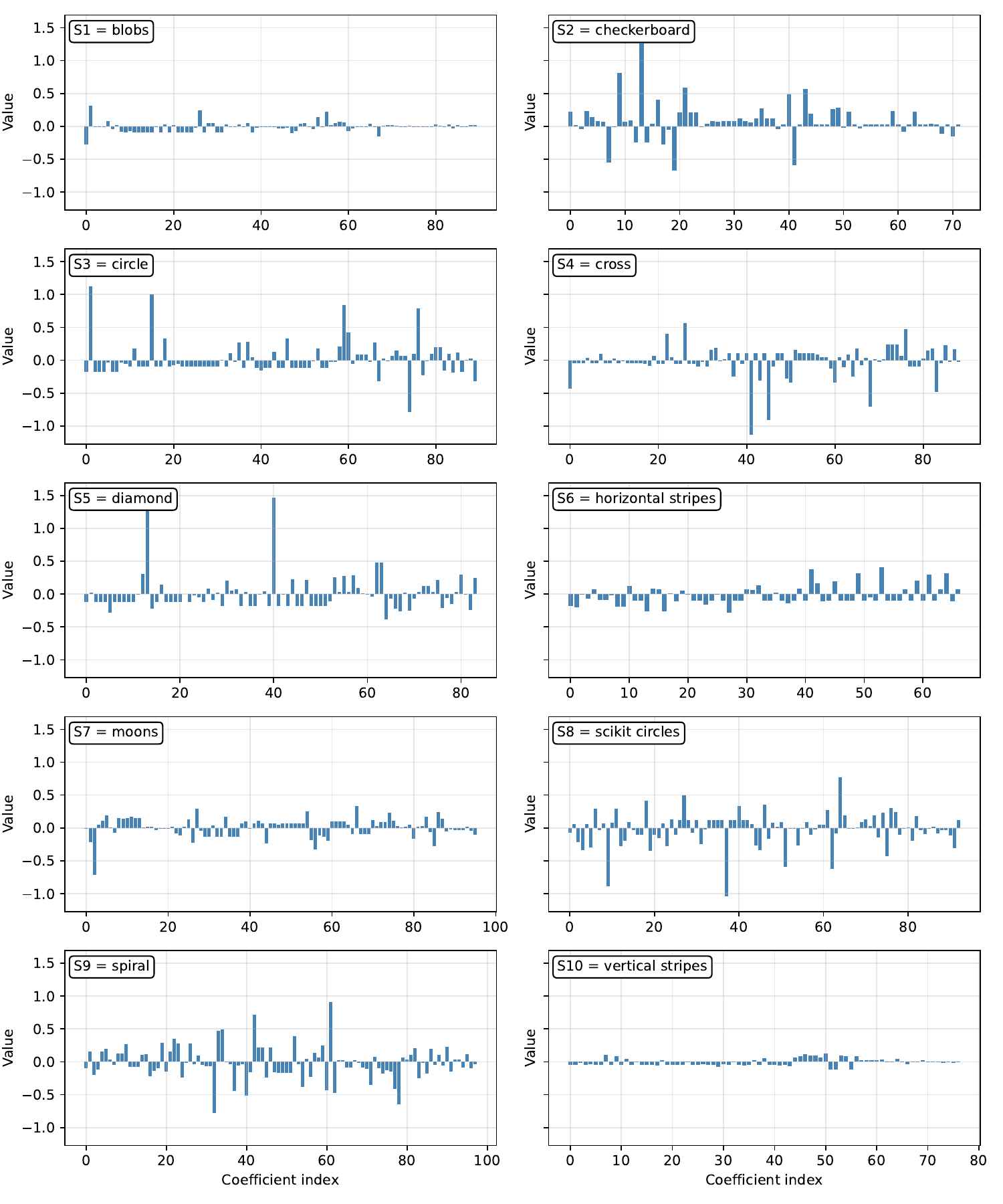}
        \caption{Generalised BVN}
        \label{fig:histograms_2d_false}
    \end{subfigure}

    \caption{Histograms of sampled coefficients for the 2D synthetic datasets using (a) the standard BVN and (b) the generalised BVN.}
    \label{fig:histograms_2d}
\end{figure}

\begin{figure}[!h]
    \centering
    \includegraphics[width=0.9\linewidth]{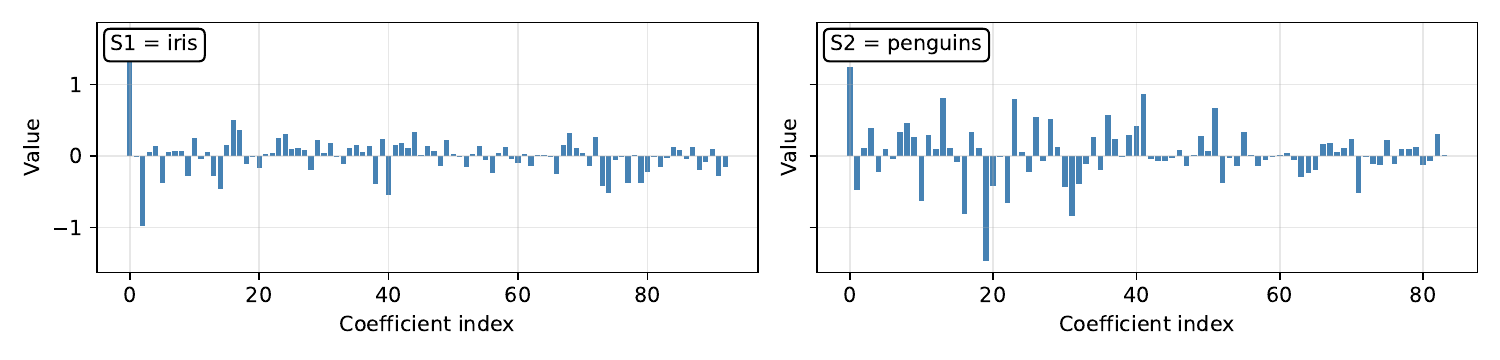}
    \includegraphics[width=0.9\linewidth]{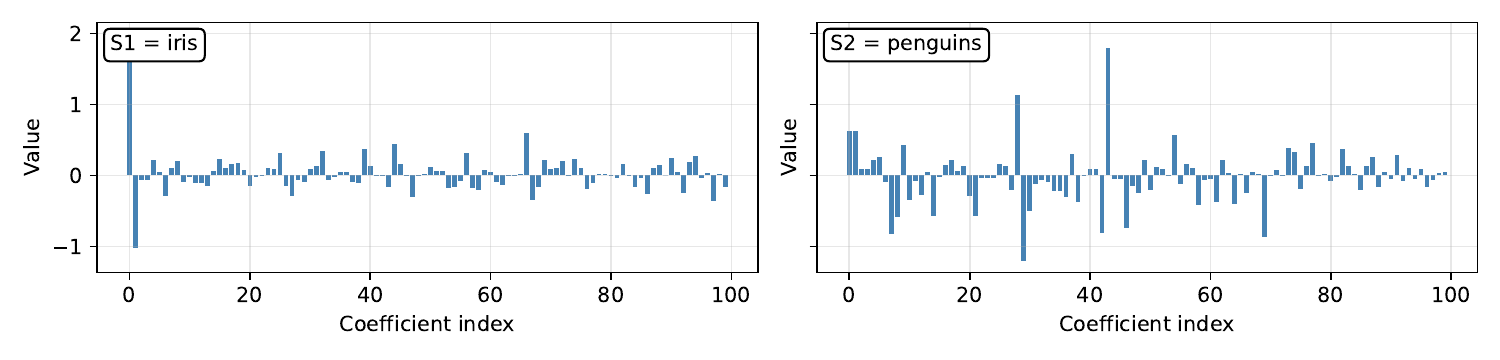}
    \caption{Histograms of the sampled coefficients obtained with the standard (top row) and generalised BV (bottom row) methods on the Iris and Penguins datasets. The coefficients are not concentrated and exhibit dispersed magnitudes.\vspace{-.7cm}}
    \label{fig:histograms_4d}
\end{figure}

\section{Comparing Filled and Unfilled Training Data}
\label{app:filled_comparison} 
As mentioned in the main text (\Cref{sec:2d_classification}), the small number of training samples ($<500$) relative to the large Hilbert-space dimension of $10^{16}$ (4 qubits per input dimension) leads to sampling leakage \cite{wakeham2024inference}, which led us to fill the dataset.
In this section, we complement the filled and unfilled comparison in \Cref{fig:4d_classification} by reporting the performance for additional values of $\lambda$.

In \Cref{fig:4d_classification_filled_comparison}, we experimentally assign a dummy fill‑label to domain inputs not present in the training set.  
This small change substantially improves the generalised model, while the standard model degrades.
The poor performance of the standard model on unfilled data is now due not to leakage but to the dummy label overwhelming the true training samples, causing the model to assign it uniformly as a constant approximation of the target function.  
The generalised model, using the same number of shots, gains enough expressivity to capture the geometry of the true training data.

\begin{figure}[!h]
    \centering

    \begin{subfigure}{0.9\linewidth}
        \centering
        \includegraphics[width=\linewidth]{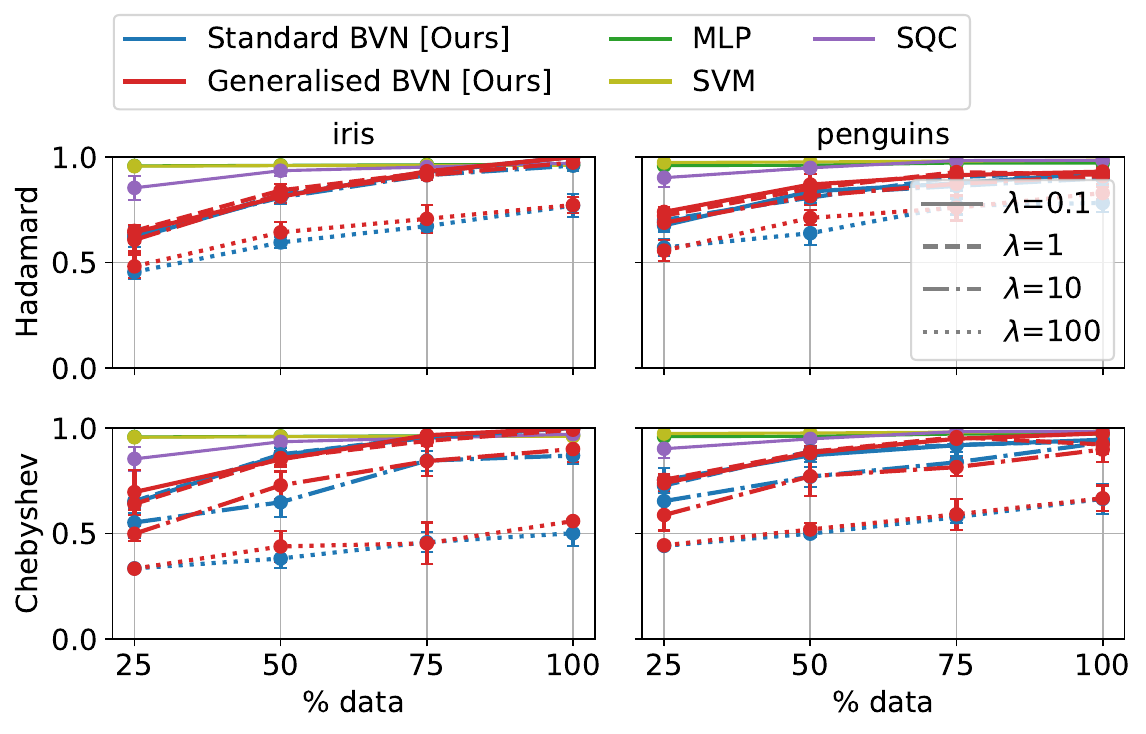}
        \caption{Unfilled training data}
        \label{fig:models_4d_unfilled}
    \end{subfigure}
    \hfill
    \begin{subfigure}{0.9\linewidth}
        \centering
        \includegraphics[width=\linewidth]{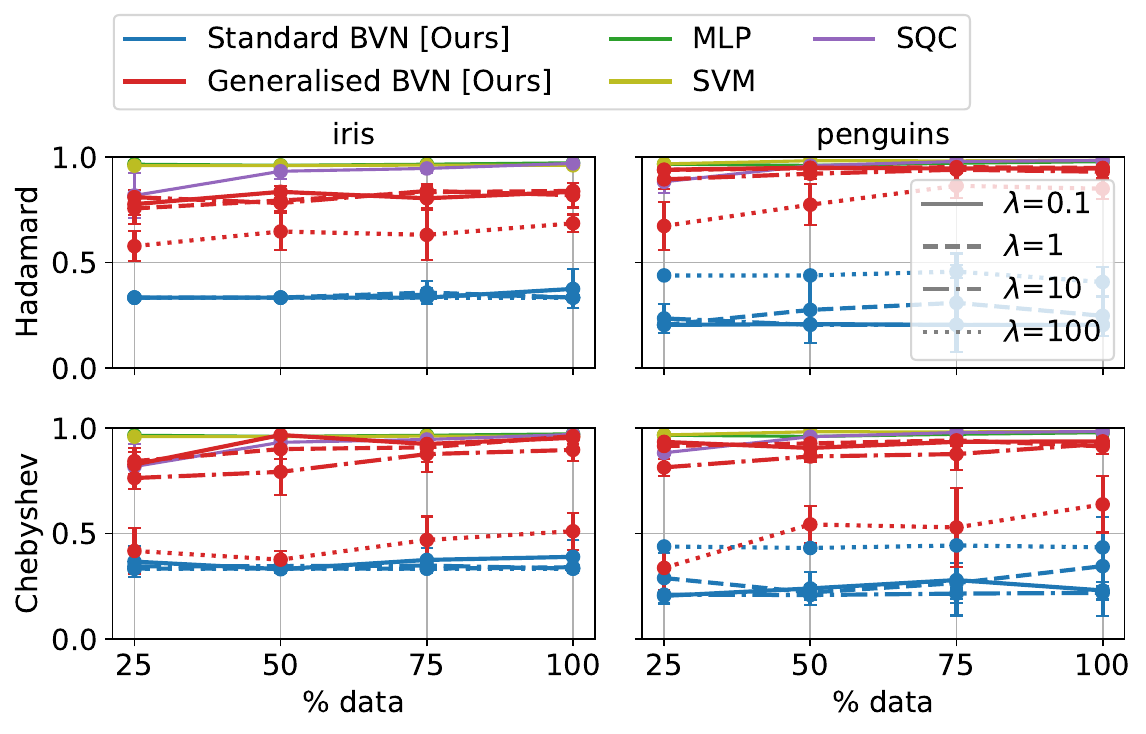}
        \caption{Filled training data}
        \label{fig:models_4d_filled}
    \end{subfigure}

    \caption{Classification results for unfilled (a) and filled (b) training data on the Iris and Penguins datasets for different $\lambda$.
    Filling the data strongly affects the performance of the BVNs:
    Unfilled data lead to poor generalisation, whereas filled data, while degrading the standard model, substantially improve the generalised model, allowing $>90\%$ accuracy with only $25\%$ training data, comparable to the baselines.
    \vspace{-5mm}
    }
    \label{fig:4d_classification_filled_comparison}
\end{figure}

\section{Future Work and Technical Extensions}
\label{app:future_work}

While the Generalised Bernstein--Vazirani model demonstrates strong capabilities on synthetic data, scaling to complex, real-world distributions (e.g., Iris, Penguins) presents challenges related to spectral leakage and shot complexity. 

The primary source of error in current experiments is the misalignment between the discrete, hard-edged quantum basis and the continuous geometry of natural data. 
We propose two enhancements to the expressive representation.

\subsection{Adaptive Basis Alignment}
Standard interference bases are static, determined solely by fixed parameters. To align the basis with the dataset's principal manifolds, we propose a hybrid variational protocol. We introduce a parametrised unitary $V(\bm{\theta})$ applied to the parameter register $|t\rangle$ prior to interference. 

The objective is to maximise the sparsity of the measurement distribution $P(y, z, t|\bm{\theta})$. We define the cost function $\mathcal{L}(\bm{\theta})$ as the Shannon entropy:
\begin{equation}
    \mathcal{L}(\bm{\theta}) = -\sum_{y,z,t} P(y,z,t|\bm{\theta}) \log P(y,z,t|\bm{\theta}).
\end{equation}
A classical pre-training loop minimises $\mathcal{L}(\bm{\theta})$, rotating the basis manifold to maximise constructive interference for the specific target dataset.

\subsection{Apodised Activation Functions}
The current rectangle activation $\chi_{rect}(x)$ acts as a boxcar function, exhibiting high-frequency leakage (Gibbs phenomenon). We propose replacing the binary thresholding with \textbf{apodised activations}. 

By replacing the multi-controlled Toffoli in the activation step with a controlled rotation $CR_y(\phi)$ on the ancilla, we can implement "soft" trapezoidal windows:
\begin{equation}
    U_{activ} |x\rangle |a\rangle = |x\rangle \left( \cos\frac{\phi(x)}{2}\mathbb{I} - \sin\frac{\phi(x)}{2}iY \right) |a\rangle,
\end{equation}
where $\phi(x)$ varies smoothly across the boundary margin $\delta$. This acts as a low-pass filter in feature space, accelerating spectral decay to $O(1/\omega^2)$ and concentrating signal energy.

\end{document}